\documentclass[lettersize,journal,twoside]{IEEEtran}
\IEEEoverridecommandlockouts
\usepackage{cite}
\usepackage{amsthm}
\usepackage{bm}
\usepackage{amsmath,amssymb,amsfonts}
\usepackage{graphicx} 
\usepackage{float} 
\usepackage{epstopdf}
\usepackage{caption}
\usepackage{algorithmic}
\usepackage{textcomp}
\usepackage{xcolor}
\usepackage{subfigure}
\usepackage[justification=centering]{caption}

\usepackage{setspace} 

\usepackage{booktabs}
\usepackage{multirow}
\usepackage{tabularx}
\usepackage{longtable}
\usepackage{supertabular}
\def\BibTeX{{\rm B\kern-.05em{\sc i\kern-.025em b}\kern-.08em
    T\kern-.1667em\lower.7ex\hbox{E}\kern-.125emX}}

\newtheorem{definition}{Definition}
\newtheorem{theorem}{Theorem}
\newtheorem{corollary}{Corollary}
\newtheorem{lemma}{Lemma}

\newcommand{\sx}{\mkern-5mu}

\begin{document}

\title{A Novel Design Method for Seeking Sparse Linear Arrays With Low Redundancy and Enhanced DOF }
\author{Si Wang and Guoqiang Xiao
\thanks{
Manuscript received 22 March 2025.
This work was supported in part by the National Natural Science Foundation of China under Grant 62371400.
The associate editor coordinating the review of this manuscript and approving it for publication was xxx.
 (Corresponding author: Guoqiang Xiao.)}
\thanks{Si Wang and Guoqiang Xiao are with the College of Computer and Information Science, Southwest University, Chongqing 400715, China
 (e-mail:wangsi@swu.edu.cn; gqxiao@swu.edu.cn).}
}

\markboth{ A Novel Design Method for Seeking Sparse Linear Arrays With Low Redundancy and Enhanced DOF,~Vol.~xx, No.~x, March~2025}%
{WNAG \MakeLowercase{\textit{et al.}}: A Novel Design Method for Seeking Sparse Linear Arrays With Low Redundancy and Enhanced DOF}

\maketitle

\begin{abstract}

Sparse arrays with $N$-sensors can provide up to $O(N^2)$ degrees of freedom (DOF) by second-order cumulants.
However, these sparse arrays like minimum-/low-redundancy arrays (MRAs/LRAs), nested arrays and coprime arrays
can only provide limited DOF and array aperture with the same number of physical sensors.
However, further increasing DOF would increase costs in practical applications.
The paper aims to design a sparse linear array (SLA) with higher DOF and lower redundancy via exploring different cases of third-order cumulants.
Based on the framework third-order exhaustive co-array (TO-ECA), a general third-order array (GTOA) with any generator is proposed in the paper.
Further, three novel arrays are designed based on GTOA with different generators,
namely third-order sum and difference array (generator) (TO-SDA(CNA)), (TO-SDA(SCNA)) and (TO-SDA(TNA-II))
which can provide closed-form expressions for the sensor locations
and enhance DOF in order to resolve more signal sources in the estimation of direction of arrival (DOA).
The three TO-SDAs are all consisted of two sub-arrays,
where the first is the generator and another is a ULA with big inter-spacing between sensors.
For the three TO-SDAs, the maximum DOF under the given number of total physical sensors can be derived
and the TO-ECA of the three TO-SDAs are hole-free.
Additionally, the redundancy of the three TO-SDAs is defined,
and the lower band of the redundancy for the three TO-SDAs is derived.
Furthermore, the proposed TO-SDA(TNA-II) not only achieves higher DOF than those of existing TONA and even SE-FL-NA
but also reduces mutual coupling effects.
Meanwhile it realizes higher resolution and decreases redundancy.
Numerical simulations are conducted to verify the superiority of TO-SDA(TNA-II) on DOA estimation performance and enhanced DOF over other existing DCAs.

\end{abstract}


\begin{IEEEkeywords}
    Sparse array, third-order cumulants, redundancy, mutual coupling, resolution, direction of arrival estimation.
\end{IEEEkeywords}

\section{Introduction}

\IEEEPARstart{L}{ow} cost sampling in intelligent perception is widely applied in many fields such as frequency estimation in the time domain \cite{Xiao2017notes}, \cite{Xiao2018robustness}, \cite{Xiao2016symmetric}, \cite{Xiao2021wrapped}, \cite{Xiao2023on}
and DOA estimation in the spatial domain.
This paper mainly focuses on DOA estimation in array signal processing,
which is a fundamental problem studied for several decades \cite{Krim1996}, \cite{Godara1997}, \cite{Tuncer2009}, \cite{Xiao22023}.
It is well known that a uniform linear array (ULA) with $N$-sensors is used to estimate ($N-1$) sources using
DOA estimation methods such as MUSIC \cite{Schmidt1986} or ESPRIT \cite{Roy1989}.
To increase the DOF of ULA, more sensors are required, thus leading to a higher cost in practical applications.
However, nonuniform linear arrays (also known as SLAs ) offer an effective solution to these problems.

For example, minimum redundancy array (MRA) \cite{Moffet1968} is a foundational structure in order to obtain as large DOF as possible
by reducing redundant sensors.
However, the non-closed form expressions for sensor positions hinder MRA's scalability and complicate large-scale array design.
This limitation leads to extensive research on nested arrays \cite{Pal2010} and coprime arrays \cite{Pal2011},
which offer significant advantages with their closed-form expressions for the sensor locations.
The success of nested arrays and coprime arrays inspires further developments aimed at enhancing DOF,
including augmented coprime arrays \cite{Pal22011},
enhanced nested arrays \cite{Zhao2019} and arrays based on the maximum element spacing criterion \cite{Zheng2019}.
With the respect of DOA estimation,
traditional subspace-based methods \cite{Liu2015} only utilize the consecutive lags of the DCA, making a hole-free configuration advantageous.
Consequently, hole-filling strategies have been proposed to create new coprime arrays-like with a hole-free DCA \cite{Wang2019}.
Obviously, various DCAs based on SOC have been widely studied in DOA estimation because of its significantly enhanced DOF \cite{Pal2010}, \cite{Pal2011}, \cite{Zhao2019}, \cite{Zheng2019}, \cite{Shi2022}.

Moreover, exploiting third-order cumulant (TOC) \cite{Sharma2023} of SLAs
can further enhance DOF \cite{Xiao2023}, starting from a mathematical model perspective.
Specifically, the corresponding third-order difference co-array (TODCA) of SLA can be obtained by calculating the TOC of the signals received based on SLA separately.
At this time, an $N$-sensors sparse array provides $\mathcal{O}(N^3)$ consecutive lags for TODCAs.

However, on the one hand, the TODCA requires an odd number of physical sensors,
which makes it more difficult to determine the physical sensor positions during array design.
Therefore, there are few references focus on the exploration of TODCA.
On the other hand, although SLAs designed based on TODCA can greatly increase the number of  consecutive lags,
using a single TOC case to design SLAs still limits the increase in the number of consecutive lags for SLAs.
Therefore, in order to increase the number of consecutive lags for virtual array obtained by using TOC to a greater extent,
a general third-order array (GTOA) with a generator is proposed based on the TO-ECA in the paper with combining four cases of TOC for received signals.
Furthermore, based on GTOA, we propose three third-order sum and difference co-array (TO-SDCA) with different generators and closed form expressions of sensor positions,
which are consisted of two sub-arrays.

In addition, the number of resolvable sources in DOA estimation is greatly affected by the DOF\cite{Piya2012}, \cite{LiuCL2016},
which are commonly adopted as indicator for quantitative evaluation and performance optimization of designing SLAs \cite{Shen2016}, \cite{Shen2019},\cite{Shen2015}.
Furthermore, the paper uses different cases of TOC to design a SLA structure with enhanced DOF, lower redundancy and smaller coupling effects.
Therefore, we first introduce the criterion of forming more consecutive lags of designing SLA as follows.

\emph{Criterion 1 (Large consecutive lags of DCA)}:
The large consecutive lags of DCA are preferred, which can not only increase the number of resolvable sources
but also lead to higher spatial resolution in the DOA estimation \cite{Pal2010}, \cite{Piya2012}, \cite{Shen2019}.

\emph{Criterion 2 (Closed-form expressions of antenna positions)}:
A closed-form expression of antenna positions is preferred for scalability considerations \cite{Cohen2019}, \cite{Cohen2020}.

\emph{Criterion 3 (Hole-free DCA)}:
A SLA with a hole-free DCA is preferred,
since the data from its DCA can be utilized directly by subspace-based DOA estimation methods
which are easy to be implemented, and thus the algorithms based on compressive sensing \cite{Shen2015}, \cite{Shen20152}, \cite{Shen2017} or co-array interpolation techniques \cite{Cui2019}, \cite{Zhou2018} with increased computational complexity can be avoided \cite{Cohen2020}, \cite{Liu20172}.

\emph{Contribution:}
This paper focuses on the design of SLA in order to get hole-free TO-ECA based on TOC with enhance DOF.
The main contributions of the paper are threefold.\par
$\bullet$ A general third-order array (GTOA) with any generator is systematically designed based on TO-ECA,
utilizing any array combining another array.
For the proposed GTOA with the given number of physical sensors,
the closed-form expressions of the physical sensor positions have been derived analytically in this paper.
\par
$\bullet$ Three novel TO-SDAs with different generators are systematically designed based on TO-ECA,
which maximizes the DOF of TO-SDAs to enhance resolving the number of sources.
For the three proposed TO-SDAs with the given number of physical sensors,
the closed-form expressions of the physical sensor positions have been derived analytically,
and the DOF of three TO-SDAs are further enhanced by improving the configuration of the physical sensors among the two ULAs.
Consequently, the three proposed TO-SDAs offers significantly higher DOF than those of other existing similar SLAs \cite{Guo2024}, \cite{Piya2012}, \cite{Shen2019}, \cite{Yang2023}.
\par
$\bullet$ The redundancy of three TO-SDAs is defined in the paper,
which is the indicator for comparing array structure performance.
Further, designing a novel array with even lower redundancy than that of SE-FL-NA.

The paper is organized as follows. In Section II, we briefly introduce the general sparse array model.
And the third-order exhaustive co-array is reviewed in Section III.
In Section IV, we describe how to design the GTOA with any generator
Furthermore, the TO-SDAs with three different generators designed based on TO-ECA are proposed in Section V,
which provide closed-form expressions for physical sensor locations by using SCA and DCA.
Additionally, we explain that the TO-ECA of the three proposed TO-SDAs is hole-free and calculate the corresponding maximum DOF.
And the redundancy of TO-SDA are defined in Section VI.
In Section VII several numerical simulations are presented to evaluate the RMSE of the three TO-SDAs
compared to TONA, FL-NA and SE-FL-NA with respect to SNR, snapshots and the number of sources.

\textit{Notations:}
$\mathbb{S}$ is the physical sensor positions set of a SLA.
$N$ is the number of antennas. $D$ is the number of source signals to be estimated. $K$ is the number of snapshots.
$\Phi$ is sensor positions set of a TO-ECA.
The operators $\otimes$, $\odot$, $(\cdot)^T$, $(\cdot)^H$ and $(\cdot)^*$ stand for the Kronecker products,
Khatri-Rao products,
transpose, conjugate transpose and complex conjugation, respectively.
Set $\{a : b : c \}$ denotes the integer line from $a$ to $c$ sampled in steps of $b\in \mathbb{N}^+$.
When $b=1$, we use shorthand $\{a : c \}$ .

\section{Preliminaries}

\subsection{General Sparse Array Model}

Assume that there are $D$ non-Gaussian and mutually uncorrelated far-field narrow band signals.
The incident angle of the $i^{th}$ signal is $\theta_i$,
and the physical sensor positions set of the SLA is represented as $\mathbb{S}=\{ p_1,p_2,...,p_N \}\cdot d$,
where the unit spacing $d$ is generally set to half wavelength.
The array output at the $n^{th}$ physical sensor corresponding to the $t^{th}$ snapshot,
denoted as $x_n(t)$, can be expressed as follows
\begin{equation}
\label{w1}
\begin{aligned}
x_n(t)=\sum_{d=1}^Da_n(\theta_d)s_d(t)+n_n(t),
\end{aligned}
\end{equation}
where $a_n(\theta_i)$ denotes the steering response of $n^{th}$ physical sensor corresponding to the $i^{th}$ source signal,
which can be expressed as follows
\begin{equation}
\label{w8}
a_n(\theta_d)=e^{j\frac{2\pi p_{l_n} d}{\lambda}sin(\theta_d)}.
\end{equation}

Further, $n_n(t)$ denotes a zero-mean additive Gaussian noise sample at the $n^{th}$ physical sensor, which is assumed to be
statistically independent of all the source signals. Thus, the received signals for all the $N$ physical sensors, which are denoted as
$\boldsymbol{x}(t)=[x_1(t),...,x_N(t)]^T$, can be written as follows
\begin{equation}
\label{w2}
\boldsymbol{x}(t)=\sum_{d=1}^D\boldsymbol{a}(\theta_d)s_d(t)+\boldsymbol{n}(t)=\boldsymbol{A}(\theta)\boldsymbol{s}(t)+\boldsymbol{n}(t),
\end{equation}
where $\boldsymbol{s}(t)=[s_1(t),...,s_D(t)]^T$ denotes the source signals vector,
$\boldsymbol{n}(t)=[n_1(t),...,n_N(t)]^T$ denotes the additive Gaussian noise vector,
and $\boldsymbol{a}(\theta_i)=[a_1(\theta_i),...,a_N(\theta_D)]^T$ denotes the array steering vector corresponding
to the $i^{th}$ source signal and $\boldsymbol{A}(\theta)=[\boldsymbol{a}(\theta_1),...,\boldsymbol{a}(\theta_D)]\in\mathbb{C}^{N\times D}$
denotes the array manifold matrix. Setting $p_1 = 0$, $\boldsymbol{a}(\theta_i)$ can be written as
\begin{equation}\nonumber
\boldsymbol{a}(\theta_d)=[1,e^{-j\frac{2\pi p_2d\sin\theta_d}{\lambda}},...,e^{-j\frac{2\pi p_Nd\sin\theta_d}{\lambda}}]^T.
\end{equation}

For $K$ numbers of snapshots, (\ref{w2}) can be rewritten in matrix form as follows
\begin{equation}
\boldsymbol{X}=\boldsymbol{A}\boldsymbol{S}+\boldsymbol{N},
\end{equation}
where $\boldsymbol{X}=[\boldsymbol{x}(1),..., \boldsymbol{x}(D)] \in \mathbb{C}^{N \times D}$ is the received signals matrix,
$\boldsymbol{S}=[\boldsymbol{s}(1),..., \boldsymbol{s}(D)] \in \mathbb{C}^{D \times K}$ is the sources signals
matrix, and $\boldsymbol{N}=[\boldsymbol{n}(1),..., \boldsymbol{n}(D)] \in \mathbb{C}^{N \times K}$ is the additive Gaussian noise matrix.

\subsection{Mutual Coupling Model}
There exists the mutual coupling effects among the physical sensors in practical applications.
When considering the effect of mutual coupling,
the received signal vector in (\ref{w2}) can be rewritten as
\begin{equation}
\label{wang5}
\boldsymbol{x}(t)=\boldsymbol{C}\boldsymbol{A}(\theta)\boldsymbol{s}(t)+\boldsymbol{n}(t),
\end{equation}
where $\boldsymbol{C}$ is the $N\times N$ mutual coupling matrix.
And, the expression for $\boldsymbol{C}$ is rather complicated generally \cite{LiuCL2016,LiuJ2017}. In the ULA configuration, $\boldsymbol{C}$ can be approximated by a
B-banded symmetric Toeplitz matrix as follows\cite{Friedlander1991}, \cite{Svantesson19991}, \cite{Svantesson19992}
\begin{equation}
\boldsymbol{C}(n_1,n_2)=
\begin{cases}
c_{|n_1-n_2|},\ \ \ \ if\ |n_1-n_2|\leq B,\\
0,\ \ \ \ \ \ \ \ otherwise,
\end{cases}
\end{equation}
where $n_1,n_2\in\mathbb{S}$, and $c_0,c_1,...,c_B$ are coupling coefficients satisfying $c_0=1>|c_1|>|c_2|>...>|c_B|$ \cite{Friedlander1991}.
For a given array with N-sensors, the coupling leakage is defined as the energy ratio \cite{Zheng2019}
\begin{equation}
\label{w24}
L=\frac{\|\boldsymbol{C}-diag(\boldsymbol{C}) \|_F}{\|\boldsymbol{C} \|_F},
\end{equation}
a small value of $L$ implies that the mutual coupling is less significant.

In addition, to lighten the notations, given any two sets $\mathbb{S}$ and $\mathbb{S}'$ \cite{Xiao2023}, we use
\begin{equation}\nonumber
\begin{aligned}
&\mathbb{C}(\mathbb{S},\mathbb{S}')=\{ p_i+p_j\ | \ p_i\in \mathbb{S},p_j\in \mathbb{S}' \},\\
&\mathbb{C}(\mathbb{S},\mathbb{S}',\mathbb{S}'')=\{ p_i+p_j+p_k\ | \ p_i\in \mathbb{S},p_j\in \mathbb{S}',p_k\in \mathbb{S}'' \},
\end{aligned}
\end{equation}
to denote the cross sum of elements from $\mathbb{S}$, $\mathbb{S}'$ and $\mathbb{S}''$.

\section{Review of Third-Order Exhaustive Co-Array}
Higher-order cumulants belong to higher-order statistics \cite{Yang2023}, which are used to describe the mathematical characteristics of stochastic processes.
When the received signals is not Gaussian, relying solely on second-order statistics cannot fully capture the statistical characteristics of the signals.
Therefore, we use higher-order cumulants, specifically TOC, to perform DOA estimation. For the zero-mean random process $\boldsymbol{x}(t)$,
the calculation expression for its TOC is as follows
\begin{equation}\nonumber
\begin{aligned}
\mathcal{C}_{3,\boldsymbol{x}}(\small{l_1 ,l_2 ,l_3})\sx =\sx cum\{ \boldsymbol{x}_{l_1}(t) ,\boldsymbol{x}_{l_2}(t) ,\boldsymbol{x}_{l_3}(t) \}\sx =\sx E\{\small{\boldsymbol{x}_{l_1}(t)\boldsymbol{x}_{l_2}(t)\boldsymbol{x}_{l_3}(t)}\},
\end{aligned}
\end{equation}
where $E\{\cdot\}$ represents the expectation operator. If the random process $\boldsymbol{x}(t)$ is a zero-mean Gaussian process, its TOC is equal to zero.

\subsection{Review of TO-ECA}
The TOC of received signal vector $\boldsymbol{x}(t)$ in (\ref{w2}) can be represented as the following four cases \cite{Sharma2023}
\begin{equation}
\label{eq1}
\begin{aligned}
&\mathcal{C}_{\boldsymbol{x}}^{(1)}\triangleq cum\{\boldsymbol{x}(t),\boldsymbol{x}(t),\boldsymbol{x}(t)\}\in\mathbb{C}^{N\times N\times N}\\[-3pt]
&=\sum\limits_{i=1}^{\small{D}}(\boldsymbol{a}(\theta_i)\otimes\boldsymbol{a}(\theta_i)\otimes\boldsymbol{a}(\theta_i))E\{ s_i(t)s_i(t)s_i(t) \},\\[-3pt]
&\mathcal{C}_{\boldsymbol{x}}^{(2)}\triangleq cum\{\boldsymbol{x}(t),\boldsymbol{x}(t),\boldsymbol{x}^*(t)\}\in\mathbb{C}^{N\times N\times N}\\[-3pt]
&=\sum\limits_{i=1}^D(\boldsymbol{a}(\theta_i)\otimes\boldsymbol{a}(\theta_i)\otimes\boldsymbol{a}^*(\theta_i))E\{ s_i(t)s_i(t)s_i^*(t) \},\\[-3pt]
&\mathcal{C}_{\boldsymbol{x}}^{(3)}\triangleq cum\{\boldsymbol{x}^*(t),\boldsymbol{x}^*(t),\boldsymbol{x}(t)\}\in\mathbb{C}^{N\times N\times N}\\[-3pt]
&=\sum\limits_{i=1}^D(\boldsymbol{a}^*(\theta_i)\otimes\boldsymbol{a}^*(\theta_i)\otimes\boldsymbol{a}(\theta_i))E\{ s_i^*(t)s_i^*(t)s_i(t) \},\\[-3pt]
&\mathcal{C}_{\boldsymbol{x}}^{(4)}\triangleq cum\{\boldsymbol{x}^*(t),\boldsymbol{x}^*(t),\boldsymbol{x}^*(t)\}\in\mathbb{C}^{N\times N\times N}\\[-3pt]
&=\sum\limits_{i=1}^D(\boldsymbol{a}^*(\theta_i)\otimes\boldsymbol{a}^*(\theta_i)\otimes\boldsymbol{a}^*(\theta_i))E\{ s_i^*(t)s_i^*(t)s_i^*(t) \}.
\end{aligned}
\end{equation}

The TOCs in (\ref{eq1}) eliminate the corresponding noise $\boldsymbol{n}(t)$,
as the cumulants of order greater than 2 are identically zero for Gaussian random processes.
The vectorization of TOC $\mathcal{C}_{\boldsymbol{x}}^{(j)}$ generates the corresponding column vector $\boldsymbol{c}_{\boldsymbol{x}}^{(j)}\in\mathbb{C}^{N^3\times1}, j\in\{1,2,3,4\}$,
as follows
\begin{equation}
\label{wh1}
\begin{aligned}
\boldsymbol{c}_{\boldsymbol{x}}^{(j)}=vec(\mathcal{C}_{\boldsymbol{x}}^{(j)})=\sum_{i=1}^Db^{(j)}(\theta_i)p_{s_i}^{(j)}=\boldsymbol{B}^{(j)}\boldsymbol{p}_{\boldsymbol{s}}^{(j)},
\end{aligned}
\end{equation}
where    $\boldsymbol{b}^{(j)}(\theta_i)\in \mathbb{C}^{N^3\times1}$, $\boldsymbol{B}^{(j)}\in\mathbb{C}^{N^3\times D}$,
$p_{s_i}^{(j)}\in\mathbb{C}$, $\boldsymbol{p}_{\boldsymbol{s}}^{(j)}\in \mathbb{C}^{D\times 1}$ .

Further, a new TO-ECA can be derived as follows by combining four $\boldsymbol{c}_{\boldsymbol{x}}^{(j)},j\in\{1,2,3,4\}$.

\begin{equation}
\label{to3}
\boldsymbol{c}_{\boldsymbol{x}}=[{\boldsymbol{c}_{\boldsymbol{x}}^{(1)}}^T,{\boldsymbol{c}_{\boldsymbol{x}}^{(2)}}^T,{\boldsymbol{c}_{\boldsymbol{x}}^{(3)}}^T, {\boldsymbol{c}_{\boldsymbol{x}}^{(4)}}^T]^T
\triangleq\boldsymbol{B}\boldsymbol{p}_s\in\mathbb{C}^{4N^3\times1},
\end{equation}
where the equivalent source signal vector $\boldsymbol{p}_s$ is expressed as follows

\begin{equation}
\label{to4}
\boldsymbol{p}_s\triangleq[{\boldsymbol{p}_s^{(1)}}^T,{\boldsymbol{p}_s^{(2)}}^T,{\boldsymbol{p}_s^{(3)}}^T, {\boldsymbol{p}_s^{(4)}}^T]^T\in\mathbb{C}^{4D\times1},
\end{equation}
and the equivalent array manifold matrix $\boldsymbol{B}$ is
\begin{equation}
\label{w21}
\boldsymbol{B}= \left(
               \begin{matrix}
                 \boldsymbol{B}^{(1)} & 0                    & 0                    &0\\
                 0                    & \boldsymbol{B}^{(2)} &0                     &0\\
                 0                    & 0                    &\boldsymbol{B}^{(3)}  &0\\
                 0                    & 0                    &0                     &\boldsymbol{B}^{(4)}\\
               \end{matrix}
             \right),
\end{equation}
where the specific expression of $\boldsymbol{B}^{(j)},\{j=1,2,3,4\}$ in (\ref{w21}) is as follows
\begin{equation}
\label{w5}
\begin{aligned}
&\boldsymbol{B}^{(j)}\triangleq[\boldsymbol{b}^{(j)}(\theta_1),\boldsymbol{b}^{(j)}(\theta_2),...,\boldsymbol{b}^{(j)}(\theta_D)],\\
&\boldsymbol{b}^{(j)}(\theta_i)\triangleq[{b}^{(j)}_1(\theta_i),{b}^{(j)}_2(\theta_i),...,{b}^{(j)}_{N^3}(\theta_i)]^T,(i=1,2,...,D)\\
&\ \ \ \ \ \ =\begin{cases}
\boldsymbol{a}(\theta_i)\otimes\boldsymbol{a}(\theta_i)\otimes\boldsymbol{a}(\theta_i), j=1,\\
\boldsymbol{a}(\theta_i)\otimes\boldsymbol{a}(\theta_i)\otimes\boldsymbol{a}^*(\theta_i),j=2,\\
\boldsymbol{a}^*(\theta_i)\otimes\boldsymbol{a}^*(\theta_i)\otimes\boldsymbol{a}(\theta_i), j=3,\\
\boldsymbol{a}^*(\theta_i)\otimes\boldsymbol{a}(\theta_i)^*\otimes\boldsymbol{a}^*(\theta_i),j=4.
\end{cases}
\end{aligned}
\end{equation}

The TOCs for the four cases of the source are as follows
\begin{equation}
\label{w4}
\begin{aligned}
&\boldsymbol{p}_{\boldsymbol{s}}^{(j)}\triangleq[C_{3,s_1(t)}^{(j)},C_{3,s_2(t)}^{(j)},...,C_{3,s_D(t)}^{(j)}]^T,\\
&\begin{matrix}
C_{3,s_i(t)}^{(j)}=\\
(i=1,2,...,D)       \\
\end{matrix}
\begin{cases}
E\{s_i(t)s_i(t)s_i(t)\},j=1,\\
E\{s_i(t)s_i(t)s_i^*(t)\},j=2,\\
E\{s_i^*(t)s_i^*(t)s_i(t)\},j=3,\\
E\{s_i^*(t)s_i^*(t)s_i^*(t)\},j=4.
\end{cases}
\end{aligned}
\end{equation}

For $\boldsymbol{c}_{\boldsymbol{x}}^{(j)}=\boldsymbol{B}^{(j)}\boldsymbol{p}_{\boldsymbol{s}}^{(j)}, j\in\{1,2,3,4\}$ given in (\ref{to3}),
it can be observed that $\boldsymbol{c}_{\boldsymbol{x}}^{(j)}$ is the result of vectorizing the TOC $\mathcal{C}_{\boldsymbol{x}}^{(j)}$.
Consequently, four virtual co-arrays under the four different cases can be obtained from
$\boldsymbol{c}_{\boldsymbol{x}}^{(j)}$,
namely first third-order co-array (TOCA$_1$), second third-order co-array (TOCA$_2$), third third-order co-array (TOCA$_3$) and fourth third-order co-array (TOCA$_4$),
which are defined as follows.

For different cases, we can get the elements of $\boldsymbol{b}^{(j)}(\theta_i), (j=1,2,3,4)$ from (\ref{w5}) as follows
\begin{equation}
 \label{w11}
 \begin{aligned}
 &b^{(1)}_{N^2(l_1-1)+N(l_2-1)+l_3}(\theta_i)=e^{j\frac{2\pi d}{\lambda}(p_{l_1}+p_{l_2}+p_{l_3})\sin(\theta_i)},\\
 &b^{(2)}_{N^2(l_1-1)+N(l_2-1)+l_3}(\theta_i)=e^{j\frac{2\pi d}{\lambda}(p_{l_1}+p_{l_2}-p_{l_3})\sin(\theta_i)},\\
 &b^{(3)}_{N^2(l_1-1)+N(l_2-1)+l_3}(\theta_i)=e^{j\frac{2\pi d}{\lambda}(-p_{l_1}-p_{l_2}+p_{l_3})\sin(\theta_i)},\\
 &b^{(4)}_{N^2(l_1-1)+N(l_2-1)+l_3}(\theta_i)=e^{j\frac{2\pi d}{\lambda}(-p_{l_1}-p_{l_2}-p_{l_3})\sin(\theta_i)}.
 \end{aligned}
 \end{equation}

Compared with the steering response $a_n(\theta_i)\sx =\sx e^{j\frac{2\pi p_{l_n} d}{\lambda}\sin(\theta_i)}$ for ULAs,
(\ref{w11}) implys the steering response of sensor located at $(p_{l_1}+p_{l_2}+p_{l_3}) d$ for TOCA$_1$, $(p_{l_1}+p_{l_2}-p_{l_3}) d$ for TOCA$_2$,
$(-p_{l_1}-p_{l_2}+p_{l_3}) d$ for TOCA$_3$ and $(-p_{l_1}-p_{l_2}-p_{l_3}) d$ for TOCA$_4$, respectively.
Consequently, the TOCA$_j, (j=1,2,3,4)$ derived from the vector $\boldsymbol{c}_{\boldsymbol{x}}^{(j)}$ in (\ref{to3})
can be considered as the virtual array, which is defined as follows.

\begin{definition}
(TOCA$_j$): For a linear array with N-sensors located at positions given by the set $\mathbb{S}$, a multiset $\Phi_j$ is defined as follows
\begin{equation}
\begin{aligned}
\label{w28}
&\Phi_1\triangleq\{(p_{l_1}+p_{l_2}+p_{l_3})d\ | \ l_1,l_2,l_3=1,2,...,N\},\\
&\Phi_2\triangleq\{(p_{l_1}+p_{l_2}-p_{l_3})d\ | \ l_1,l_2,l_3=1,2,...,N\},\\
&\Phi_3\triangleq\{(-p_{l_1}-p_{l_2}+p_{l_3})d\ | \ l_1,l_2,l_3=1,2,...,N\},\\
&\Phi_4\triangleq\{(-p_{l_1}-p_{l_2}-p_{l_3})d\ | \ l_1,l_2,l_3=1,2,...,N\},
\end{aligned}
\end{equation}
where the multiset $\Phi_j$ allows repetitions, and has an underlying set $\Phi_j^{u}$ that contains the unique elements of $\Phi_j$.
Consequently, TOCA$_j$ is defined as the virtual linear array for the four cases,
where the sensors are located at positions given by the set $\Phi_j^u$.
\end{definition}

Furthermore, $\boldsymbol{c}_{\boldsymbol{x}}\in\mathbb{C}^{4N^2\times1}$ in (\ref{to3}) is obtained by
combing four $\boldsymbol{c}_{\boldsymbol{x}}^{(j)}, j\in\{1,2,3,4\}$.
It is equivalent to the cumulants of the signals received in a single snapshot by the constructed virtual linear array,
which is the combination of all four possible co-arrays for $j\in \{1, 2, 3, 4\}$, namely TOCA$_1$, TOCA$_2$, TOCA$_3$ and TOCA$_4$.
Therefore, the derived virtual linear array is the TO-ECA,
and to obtain the sensor position set of TO-ECA, we first introduce the following knowledge about multiset.

A multiset $\Phi$ is defined as the multiset-sum (bag sum) of the four multisets $\Phi= \Phi_1 + \Phi_2 + \Phi_3+\Phi_4$,
which denotes the union operation of the set subjected to sets with duplicate elements, and the specific information is in \cite{Wang2024}.
On the contrary, for the multiset $\Phi$, there exists an set $\Phi^u=\Phi_1^u \cup \Phi_2^u \cup \Phi_3^u\cup \Phi_4^u$ that contains the unique elements.

After obtaining the sensor position set of TO-ECA, the TO-ECA is defined as follows.
\begin{definition}
(TO-ECA):
For a linear array of N-sensors located at positions given by the set $\mathbb{S}$,
the TO-ECA is derived based on the TOCs, whose sensors are located at the set $\Phi^u$.
\end{definition}

Therefore, the TO-ECA is defined as the virtual linear array contained $|\Phi^u|=\mathcal{O}(4N^2)$ sensors.

\subsection{Proposed TO-ECA With Mutual Coupling}
The TO-ECA with mutual coupling can be derived based on the mutual coupling model in (\ref{wang5}) as follows
\begin{equation}
\label{eq2}
\begin{aligned}
\boldsymbol{\tilde{z}}_{\boldsymbol{x}}^{(1)}=\boldsymbol{C}^{(1)}_{vec}\boldsymbol{B}^{(1)}\boldsymbol{p}_{\boldsymbol{s}}^{(1)},\
\boldsymbol{\tilde{z}}_{\boldsymbol{x}}^{(2)}=\boldsymbol{C}^{(2)}_{vec}\boldsymbol{B}^{(2)}\boldsymbol{p}_{\boldsymbol{s}}^{(2)},\\
\boldsymbol{\tilde{z}}_{\boldsymbol{x}}^{(3)}=\boldsymbol{C}^{(3)}_{vec}\boldsymbol{B}^{(3)}\boldsymbol{p}_{\boldsymbol{s}}^{(3)},\
\boldsymbol{\tilde{z}}_{\boldsymbol{x}}^{(4)}=\boldsymbol{C}^{(4)}_{vec}\boldsymbol{B}^{(4)}\boldsymbol{p}_{\boldsymbol{s}}^{(4)},
\end{aligned}
\end{equation}
further, combining the four $\boldsymbol{\tilde{z}}_{\boldsymbol{x}}^{(j)},j\in\{1,2,3,4\}$ to
derive TO-ECA with mutual coupling
\begin{equation}
\begin{aligned}
\label{w20}
\boldsymbol{\tilde{z}}_{\boldsymbol{x}}&=[{\boldsymbol{\tilde{z}}_{\boldsymbol{x}}^{(1)}}^T,{\boldsymbol{\tilde{z}}_{\boldsymbol{x}}^{(2)}}^T,
{\boldsymbol{\tilde{z}}_{\boldsymbol{x}}^{(3)}}^T,{\boldsymbol{\tilde{z}}_{\boldsymbol{x}}^{(4)}}^T]^T
\triangleq\boldsymbol{\tilde{C}}_{vec}\boldsymbol{B}\boldsymbol{p}_s\in\mathbb{C}^{4N^2\times1},
\end{aligned}
\end{equation}
where $\boldsymbol{B}$ and $\boldsymbol{p}_{\boldsymbol{s}}$ are shown in (\ref{w5}) and (\ref{w4}), respectively.
And the virtual mutual coupling matrix $\boldsymbol{\tilde{C}}_{vec}\in\mathbb{C}^{4N^2\times4N^2}$ of TO-ECA is shown as follows

\begin{equation}\nonumber
\begin{cases}
\boldsymbol{C}^{(1)}_{vec}=\boldsymbol{C}\otimes\boldsymbol{C}\otimes\boldsymbol{C},\
\boldsymbol{C}^{(2)}_{vec}=\boldsymbol{C}\otimes\boldsymbol{C}\otimes\boldsymbol{C}^*,\\
\boldsymbol{C}^{(3)}_{vec}=\boldsymbol{C}^*\otimes\boldsymbol{C}^*\otimes\boldsymbol{C},\
\boldsymbol{C}^{(4)}_{vec}=\boldsymbol{C}^*\otimes\boldsymbol{C}^*\otimes\boldsymbol{C}^*,
\end{cases}
\end{equation}

\begin{equation}\nonumber
\boldsymbol{\tilde{C}}_{vec}= \left(
               \begin{matrix}
                 \boldsymbol{C}^{(1)}_{vec} & 0                          & 0                         &0\\
                 0                          & \boldsymbol{C}^{(2)}_{vec} &0                          &0\\
                 0                          & 0                          &\boldsymbol{C}^{(3)}_{vec} &0\\
                 0                          & 0                          &0                          &\boldsymbol{C}^{(4)}_{vec} \\
               \end{matrix}
             \right).
\end{equation}


\section{general design method for third-order array}
In this section, we present a general method to design general third-order array (GTOA) based on TO-ECA by using a generator.
The generator is the basic construction block of the consecutive virtual array generating by GTOA, and it can be any array configuration of choice.

\begin{definition}
(GTOA with generator $\mathcal{G}$): The physical sensor positions of the $\mathcal{G}-\mathcal{H}$ are given by
\begin{equation}
\text{GTOA}\triangleq \mathcal{G}\cup \mathcal{H},
\end{equation}
where $\mathcal{H}=\{ \delta_1+\delta_2\mathcal{F} \}$, $\delta_1,\delta_2\in \mathbb{N}$,$\mathcal{F}=\{0,1,2,...,N_2-1\}$.
And the structure of the GTOA with generator $\mathcal{G}$ is shown as Fig. 1.
\end{definition}

\begin{figure}[h]
 \center{\includegraphics[width=9cm]  {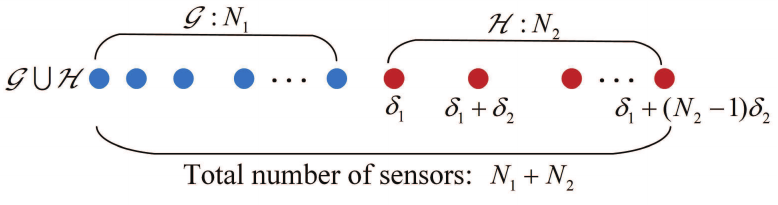}}
 \caption{\label{1} Structure of GTOA}
\end{figure}

\section{Third-order sum-difference co-array designs based on the generator}
When the TO-ECA is hole-free, it can be easily utilized to estimate DOA without any spatial aliasing \cite{Pal2010}, \cite{Pal22011}, \cite{Piya2012}.
Moreover, the number of consecutive lags of DCA mainly
depends on the physical sensor geometry of a linear array \cite{Cohen2020}.
Therefore, the third-order sum and difference array (TO-SDA) with different generators is proposed based on TO-ECA in the part to enhance the DOF and reduce the redundancy.
The TO-SDA with the specific three different generators is designed systematically by appropriately deploying the physical sensor positions of two sub-arrays as shown in Fig. 2,
where $\mathcal{G}$ is the generator, and $\mathcal{H}$ is a ULA with a big inter-spacing between sensors.
\subsection{Generator With CNA}
\begin{definition}
(TO-SDA(CNA)) The TO-SDA(CNA) consists of two sub-arrays with the number of physical sensors $N=N_1+N_2$,
where $N_1$ and $N_2$ represent the number of physical sensors in $\mathcal{G}$ and $\mathcal{H}$.
These sensors in TO-SDA(CNA) are located at positions given by the set $\mathcal{G}$ and $\mathcal{H}$, respectively,
which can be represented as follows

\begin{equation}
\label{wh1}
\begin{aligned}
&\mathbb{S}=\mathcal{G}\cup\mathcal{H},\\
&\mathcal{G}=\{ 0:M_1-1 \}\sx \cdot\sx d\ \cup\\
&\ \ \ \ \{ M_1:M_1+1:M_1+(M_1+1)(M_2-1) \}\sx \cdot \sx d\ \cup\\
&\ \ \ \ \{M_1\sx +\sx (M_1\sx +\sx 1)(M_2\sx -\sx 1)\sx +\sx 1 \sx : \sx 2M_1\sx +\sx (M_1 \sx + \sx 1)(M_2\sx -\sx 1) \}\sx \cdot\sx d,\\
&\mathcal{H}=\{ \delta_1+\delta_2\mathcal{F} \},\ \mathcal{F}=\{0,1,2,...,N_2-1\}\cdot d,\\
&0 \leq \delta_1 \leq \lambda_1+\lambda_2+1,\ 0 \leq \delta_2 \leq 2\lambda_1+1,\\
&\lambda_1\sx =\sx 2(M_1\sx -\sx 1)\sx +\sx 2M_2(M_1\sx +\sx 1), \lambda_2\sx =\sx 3(M_1\sx -\sx 1)\sx +\sx 3M_2(M_1\sx +\sx 1),
\end{aligned}
\end{equation}
where $\lambda_1$ and $\lambda_2$  are the length of longest consecutive segment in second-order SCA and third-order SCA, respectively.
The structure of the TO-SDA(CNA) is shown as Fig. 2 (a).
\end{definition}

\begin{figure*}
\label{example1}
  \centering
  \subfigure[]{
    \label{fig:subfig:onefunction}
    \includegraphics[scale=0.235]{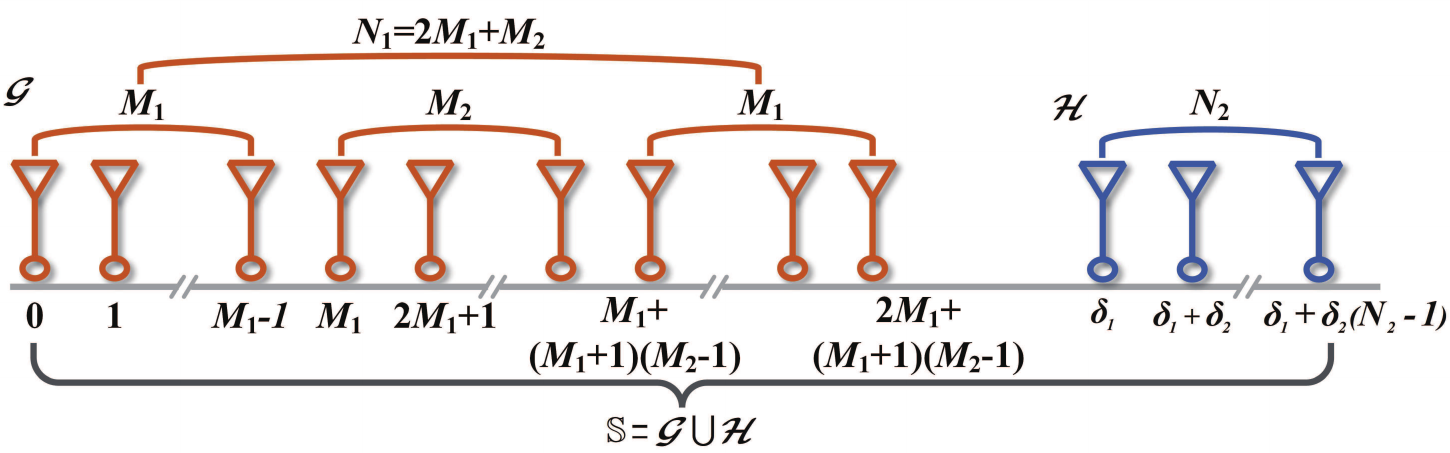}}
  \hspace{0in} 
  \subfigure[]{
    \label{fig:subfig:threefunction}
    \includegraphics[scale=0.235]{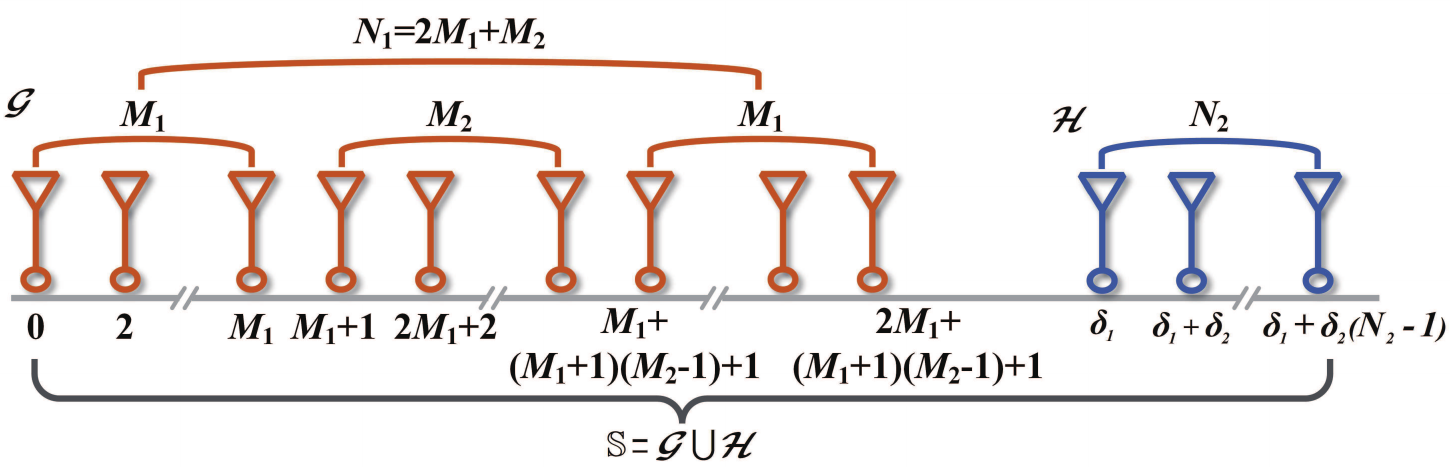}}
    \hspace{0in} 
  \subfigure[]{
    \label{fig:subfig:threefunction}
    \includegraphics[scale=0.235]{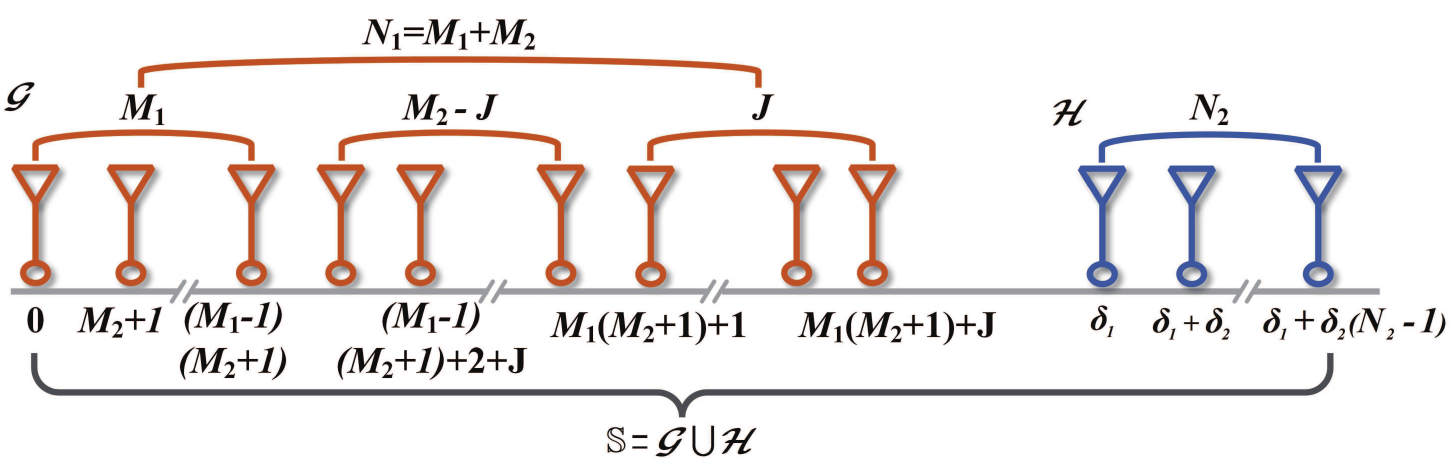}}
  \caption{Structure of TO-SDA with different generators
 (a) TO-SDA(CNA). (b) TO-SDA(SCNA). (c) TO-SDA(TNA-II) }
\end{figure*}

\subsubsection{Consecutive Lags of the TO-SDA(CNA)}

\begin{lemma}
The TO-ECA of TO-SDA(CNA) is hole-free.
\end{lemma}

\begin{proof}
Firstly, we can get the second-order and third-order cross sum and difference of $\mathcal{G}$ as follows
\begin{equation}
\label{to5}
\begin{aligned}
&\mathbb{C}(\mathcal{G},\mathcal{G})=\{ 0: \lambda_1\},\
\mathbb{C}(-\mathcal{G},-\mathcal{G})=\{ -\lambda_1:0\},\\
&\mathbb{C}(\mathcal{G},\mathcal{G},\mathcal{G})=\{ 0:\lambda_2\},\
\mathbb{C}(-\mathcal{G},-\mathcal{G},-\mathcal{G})=\{ -\lambda_2:0\}.
\end{aligned}
\end{equation}

Thus, the union of above third-order cross sum and difference is
\begin{equation}\nonumber
\begin{aligned}
&\mathbb{C}(\mathcal{G},\mathcal{G},\mathcal{G})\cup
\mathbb{C}(-\mathcal{G},-\mathcal{G},-\mathcal{G})=\{ -\lambda_2:\lambda_2\}.
\end{aligned}
\end{equation}

Secondly, when $\delta_1=\lambda_1+\lambda_2+1$, $\delta_2=2\lambda_1+1$, the $\mathcal{H}$ is
\begin{equation}\nonumber
 \mathcal{H}\supseteq\mathcal{H}^*=\{ \lambda_1+\lambda_2+1 \mkern-5mu : \mkern-5mu 2\lambda_1+1 \mkern-5mu : \mkern-5mu \lambda_1+\lambda_2+1+(N_2-1)(2\lambda_1+1) \} .
\end{equation}

The cross sum and difference of $\mathbb{C}(\mathcal{G},\mathcal{G})$ and $\mathcal{H}^*$ can be obtained as follows
\begin{equation}
\label{to6}
\begin{aligned}
&\mathbb{C}(\mathcal{H}^*,-\mathbb{C}(\mathcal{G},\mathcal{G}))=\{ \lambda_2\sx +\sx 1\sx :\sx \lambda_1\sx +\sx \lambda_2\sx +\sx 1 \sx +\sx (N_2\sx -\sx 1)(2\lambda_1\sx +\sx 1)\},\\
&\mathbb{C}(\mathcal{H}^*,\mathbb{C}(\mathcal{G},\mathcal{G}))=\{ \lambda_1\sx +\sx \lambda_2\sx +\sx 1\sx :\sx 2\lambda_1\sx +\sx \lambda_2\sx +\sx 1
\sx +\sx (N_2\sx -\sx 1)(2\lambda_1\sx +\sx 1))\},\\
&\mathbb{C}(\sx -\sx \mathcal{H}^*\sx ,\sx -\sx \mathbb{C}(\mathcal{G} , \mathcal{G}))\sx =\sx \{ -2\lambda_1\sx -\sx \lambda_2\sx -\sx 1\sx -\sx (N_2\sx -\sx 1)(2\lambda_1\sx +\sx 1)\sx :\sx -\sx \lambda_1\sx -\sx \lambda_2\sx -\sx 1\},\\
&\mathbb{C}(-\mathcal{H}^*,\mathbb{C}(\mathcal{G},\mathcal{G}))=\{ -\lambda_1\sx -\sx \lambda_2\sx -\sx 1\sx -\sx (N_2\sx -\sx 1)(2\lambda_1\sx +\sx 1)\sx :\sx -\sx \lambda_2\sx -1\}.
\end{aligned}
\end{equation}

$\lambda_1+\lambda_2+1+(N_2-1)(2\lambda_1+1)\geq\lambda_1+\lambda_2+1$ when $N_2\geq1$.
Further, the union of above sets in (\ref{to5}) and (\ref{to6}) is
 \begin{equation}\nonumber
\begin{aligned}
&\mathbb{C}(\mathcal{H}^*,-\mathbb{C}(\mathcal{G},\mathcal{G}))\cup
\mathbb{C}(\mathcal{H}^*,\mathbb{C}(\mathcal{G},\mathcal{G}))\cup
\mathbb{C}(-\mathcal{G},-\mathcal{G})\cup\\
&\mathbb{C}(\mathcal{G},\mathcal{G})\cup
\mathbb{C}(-\mathcal{H}^*,-\mathbb{C}(\mathcal{G},\mathcal{G}))\cup
\mathbb{C}(-\mathcal{H}^*,\mathbb{C}(\mathcal{G},\mathcal{G}))\\
&=\{  -2\lambda_1-\lambda_2-1-(N_2-1)(2\lambda_1+1):\\[-2pt]
&\ \ \ \ 2\lambda_1+\lambda_2+1+(N_2-1)(2\lambda_1+1) \}.
\end{aligned}
\end{equation}

To sum up, the TO-ECA of TO-SDA(CNA) is hole-free, and the total consecutive lags of TO-ECA are
$4\lambda_1+2\lambda_2+2(N_2-1)(2\lambda_1+1)+3=(6+8N_2)[(M_1-1)+M_2(M_1+1)]+2N_2+1$.

\end{proof}

\subsubsection{The Maximum DOF of TO-SDA(CNA) With the Given Number of Physical Sensors}
~\par
The DOF of TO-SDA(CNA) designed based on TO-ECA can be further increased by optimizing the distribution of the physical sensors among
$\mathcal{G}$ and $\mathcal{H}$ for a given number of physical sensors.

\begin{lemma}
To obtain the maximum DOF of TO-SDA(CNA) with the given number of sensors, the number of sensors in $\mathcal{G}$ and $\mathcal{H}$ is set to
\begin{equation}\nonumber
\begin{cases}
N_1=\lceil \frac{4N+\sqrt{(4N+15)^2+672}-21}{12} \rfloor,\\
N_2=N-N_1,\ N_1=2M_1+M_2\\[-1pt]
M_1=\lceil\frac{N_1-1}{4}\rfloor, M_2=N_1-2\lceil\frac{N_1-1}{4}\rfloor.
\end{cases}
\end{equation}
\end{lemma}

\begin{proof}
See Appendix A.
\end{proof}

\subsection{Generator With Shifted CNA}

\begin{definition}
(TO-SDA(SCNA)) The TO-SDA(SCNA) consists of two sub-arrays with the number of physical sensors $N=N_1+N_2$,
where $N_1$ and $N_2$ represent the number of physical sensors in $\mathcal{G}$ and $\mathcal{H}$.
These sensors in TO-SDA(SCNA) are located at positions given by the set $\mathcal{G}$ and $\mathcal{H}$, respectively,
which can be represented as follows
\begin{equation}
\label{wh4}
\begin{aligned}
&\mathbb{S}=\mathcal{G}\cup\mathcal{H},\\
&\mathcal{G}=\{ 0 \} \cup  \{ 2:M_1 \}\cdot d\ \cup\\
&\ \ \ \ \{ M_1+1:M_1+1: M_2(M_1+1) \}\cdot d\ \cup\\
&\ \ \ \ \{M_2(M_1+1)+1:2M_1+(M_1+1)(M_2-1)+1 \}\cdot d,\\
&\mathcal{H}=\{ \delta_1+\delta_2\mathcal{F} \},\ \mathcal{F}=\{0,1,2,...,N_2-1\}\cdot d,\\
&0 \leq \delta_1 \leq \lambda_1+\lambda_2+1,\ 0 \leq \delta_2 \leq 2\lambda_1+1,\\
&\lambda_1\sx =\sx 2(2M_1\sx +\sx M_2)\sx +\sx 2M_1(M_2\sx -\sx 1), \lambda_2\sx =\sx 3(2M_1\sx +\sx M_2)\sx +\sx 3M_1(M_2\sx -\sx 1),
\end{aligned}
\end{equation}
where $\lambda_1$ and $\lambda_2$  are the length of longest consecutive segment in second-order SCA and third-order SCA, respectively.
The structure of the TO-SDA(SCNA) is shown as Fig. 2 (b).
\end{definition}

\subsubsection{Consecutive Lags of the TO-SDA(SCNA)}

\begin{lemma}
The TO-ECA of TO-SDA(SCNA) is hole-free.
\end{lemma}

\begin{proof}
Based on Lemma 1, it can be similarly proven that Lemma 3 holds true.
\end{proof}

\subsubsection{The Maximum DOF of TO-SDA(SCNA) With the Given Number of Physical Sensors}
~\par
The DOF of TO-SDA(SCNA) designed based on TO-ECA can be further increased by optimizing the distribution of the physical sensors among
$\mathcal{G}$ and $\mathcal{H}$ for a given number of physical sensors.

\begin{lemma}
To obtain the maximum DOF of TO-SDA(SCNA) with the given number of sensors, the number of sensors in $\mathcal{G}$ and $\mathcal{H}$ is set to
\begin{equation}\nonumber
\begin{cases}
N_1=\lceil \frac{4N+\sqrt{(4N+15)^2+288}-21}{12} \rfloor,\\
N_2=N-N_1,\ N_1=2M_1+M_2\\
M_1=\lceil\frac{N_1-1}{4}\rfloor, M_2=N_1-2\lceil\frac{N_1-1}{4}\rfloor.
\end{cases}
\end{equation}
\end{lemma}

\begin{proof}
See Appendix B.
\end{proof}

\subsection{Generator With TNA-II}
\begin{definition}
(TO-SDA(TNA-II)) The TO-SDA(TNA-II) consists of two sub-arrays with the number of physical sensors $N=N_1+N_2$,
where $N_1$ and $N_2$ represent the number of physical sensors in $\mathcal{G}$ and $\mathcal{H}$.
These sensors in TO-SDA(TNA-II) are located at positions given by the set $\mathcal{G}$ and $\mathcal{H}$, respectively,
which can be represented as follows
\begin{equation}
\label{to9}
\begin{aligned}
&\mathbb{S}=\mathcal{G}\cup\mathcal{H},\\
&\mathcal{G}=\mathbb{L}_1\cup\mathbb{L}_2\cup\mathbb{L}_3,\\
&\mathbb{L}_1=\{ 0:M_1+1:(M_1-1)(M_2+1) \}\cdot d\,\\
&\mathbb{L}_2=\{ (M_1 \mkern-5mu - \mkern-5mu 1)(M_2\mkern-5mu +\mkern-5mu 1)+ \mkern-5mu J+1 \mkern-5mu :\mkern-5mu (M_1\mkern-5mu -\mkern-5mu 1)(M_2\mkern-5mu +\mkern-5mu 1)\mkern-5mu +\mkern-5mu M_2 \}\cdot d,\\
&\mathbb{L}_3=\{ M_1(M_2+1)+1:M_1(M_2+1)+J\}\cdot d,\\
&\mathcal{H}\sx =\sx \{ \delta_1+\delta_2\mathcal{F} \}, \mathcal{F}\sx =\sx \{0,1,2,...,N_2-1\}\cdot d, J\sx =\sx \lceil N_1/2 \rceil\sx -\sx 1, \\
&0 \leq \delta_1 \leq \lambda_1+\lambda_2+1,\ 0 \leq \delta_2 \leq 2\lambda_1+1,
\end{aligned}
\end{equation}
where $\lambda_1$ and $\lambda_2$  are the length of longest consecutive segment in second-order SCA and third-order SCA, respectively.
The structure of the TO-SDA(TNA-II) is shown as Fig. 2 (c).
\end{definition}

\subsubsection{Consecutive Lags of the TO-SDA(TNA-II)}

\begin{lemma}
The TO-ECA of TO-SDA(TNA-II) is hole-free.
\end{lemma}

\begin{proof}
Based on Lemma 1, it also can be similarly proven that Lemma 5 holds true.
\end{proof}

~\par

\subsubsection{The Maximum DOF of TO-SDA(TNA-II) With the Given Number of Physical Sensors}
~\par
The DOF of TO-SDA(TNA-II) designed based on TO-ECA can be further increased by optimizing the distribution of
the physical sensors among $\mathcal{G}$ and $\mathcal{H}$ for a given number of physical sensors.

\begin{lemma}
To obtain the maximum DOF of TO-SDA(TNA-II) with the given number of sensors, the number of sensors in $\mathcal{G}$ and $\mathcal{H}$ is set to
\begin{equation}\nonumber
\begin{cases}
N_1=\frac{3+4N+\sqrt{(4N-9)^2+36}}{12},\\
N_2=N-N_1,\ N_1=M_1+M_2,\ J=\lceil \frac{N_1}{2} \rceil-1,\\
M_1=N_1-\lceil \frac{2N_1-1}{4} \rfloor,\ M_2=\lceil \frac{2N_1-1}{4} \rfloor.
\end{cases}
\end{equation}
\end{lemma}

\begin{proof}
See Appendix C.
\end{proof}

\subsection{Example of the Proposed Arrays Based on Different Generator}

\begin{figure}
\label{example1}
  \centering
  \subfigure[]{
    \label{fig:subfig:onefunction}
    \includegraphics[scale=0.24]{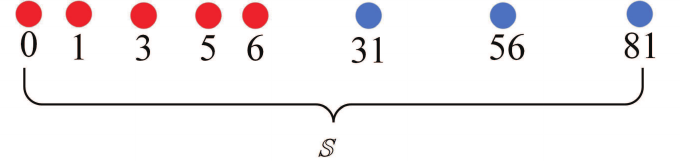}}
  \hspace{0in} 
  \subfigure[]{
    \label{fig:subfig:threefunction}
    \includegraphics[scale=0.21]{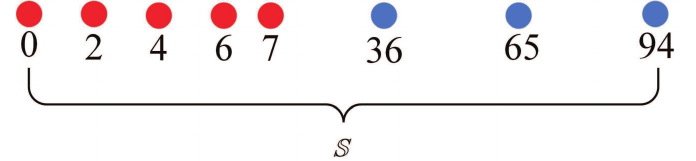}}
    \hspace{0in} 
  \subfigure[]{
    \label{fig:subfig:threefunction}
    \includegraphics[scale=0.191]{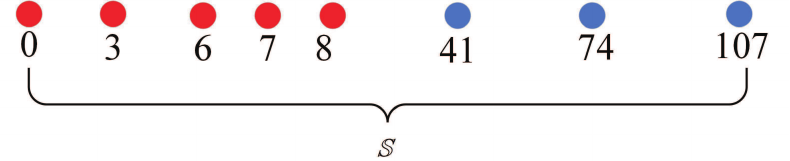}}
  \caption{An example for TO-SDA with different generator while $N=8$
 (a) TO-SDA(CNA). (b) TO-SDA(SCNA). (c) TO-SDA(TNA-II) }
\end{figure}

The array structures of TO-SDA(CNA), TO-SDA(SCNA) and TO-SDA(TNA-II) designed based on TO-ECA with $N=8$
physical sensors are shown in Fig. 3.
The DOF of TO-SDA(CNA), TO-SDA(SCNA) and TO-SDA(TNA-II) are $187$, $217$ and $247$, which are higher than the DOF of $163$ for FL-NA and SE-FL-NA in \cite{Sharma2023}.

\section{Redundancy}
In order to define the redundancy of TO-SDA with different generators, the redundancy of DCA and SCA are introduced as follows
\begin{definition}
(Redundancy\cite{Hoctor1990}). The redundancy of a N-antennas SLA with a consecutive sum co-array is
\begin{equation}\nonumber
R_S=\frac{N(N+1)/2}{2N+1}\geq 1.
\end{equation}
\end{definition}

\begin{definition}
(Redundancy \cite{Moffet1968}). The redundancy of a N-sensors SLA with a consecutive difference co-array is
\begin{equation}\nonumber
R_D=\frac{N(N-1)/2}{E}\geq 1,
\end{equation}
where $E$ is the aperture of the consecutive difference co-array.
\end{definition}
\subsection{The Size of the TO-ECA}
We first study the size of TO-ECA, which not only bounds the size of consecutive segment,
but also in effect characterizes an upper bound of the number of identifiable sources by any DOA method using
third-order cumulant. Before presenting the bounds for the size of $\Phi^u$,
we introduce the definitions of second-order DCA $\mathbb{D}_2$ and SCA $\mathbb{S}_2$ for array with $N$-sensors located at set $\mathbb{S}$,
which is helpful in the following development.
\begin{equation}\nonumber
\begin{aligned}
\mathbb{D}_2=\{l_{n_1}-l_{n_2}|n_1,n_2\in[0,N-1]\},\\
\mathbb{S}_2=\{l_{n_1}+l_{n_2}|n_1,n_2\in[0,N-1]\}.
\end{aligned}
\end{equation}

 It has been shown that
\begin{equation}\nonumber
\begin{aligned}
&2N-1\leq | \mathbb{D}_2 | \leq N^2-N+1,\\
&2N-1\leq | \mathbb{S}_2 | \leq N^2+N+1.
\end{aligned}
\end{equation}\\[-22pt]
\begin{theorem}
The size of $\Phi^u$ has the following bounds
\begin{equation}\nonumber
6N-5 \leq | \Phi^u | \leq k(N),
\end{equation}
where the upper bound $k(N)$ is defined as
\begin{equation}\nonumber
k(N):=\frac{4N^3+3N^2-N+3}{3}.
\end{equation}

\end{theorem}

\begin{proof}
For the lower bound, it is known \cite{Pal2010} that a linear array with $N$-sensors has at least $2N-1$ distinct sensors in its DCA.
Similarly, it can be inferred that a linear array with $N$-sensors has at least $6N-5$ distinct sensors in its TO-ECA \cite{Sharma2023}.

As for the upper bound, since TO-ECA is symmetric corresponding to zero, the positive segment of TO-ECA is considered as follows.
For the case $\{l_{n_1}+l_{n_2}+l_{n_3} | (l_{n_1},l_{n_2},l_{n_3}\in[0,N-1])\}$ of TO-ECA,
there are total number of $1+(1+2)+(1+2+3)+\cdots+(1+2+3+\cdots+(N-1))=\frac{(N-1)N(N+1)}{6}$.
For the case $\{l_{n_1}-l_{n_2}-l_{n_3} | (l_{n_1},l_{n_2},l_{n_3}\in[0,N-1])\}$ of TO-ECA,
there are total number of $4\frac{N(N+1)}{2}$.
Therefore, the positive segment of TO-ECA is $\frac{(N-1)N(N+1)}{6}+4\sx \cdot\sx \frac{N(N+1)}{2}=\frac{4N^3+3N^2-N}{6}$,
and the upper bound of TO-ECA is $k(N)=\frac{4N^3+3N^2-N+3}{3}$.

\end{proof}

\subsection{Redundancy of TO-ECA}
After obtaining the size of consecutive segment for TO-ECA $\mathbb{Z}$,
and according to the redundancy definitions of DCA and SCA, it can be seen that $2N+1$ and $2L+1$ are the consecutive segment in SCA and DCA.
Further, it is can be known that the bounds of $R_S$ and $R_D$ are characterized the size of $N$ and $L$.
Therefore, the TO-ECA redundancy is defined as follows by generalizing the redundancy definition idea of SCA and DCA.\\[-15pt]

\begin{definition}
The redundancy of TO-ECA is defined as\\[-5pt]
\begin{equation}\nonumber
R_{T}=\frac{\tilde{k}(N)}{Z},
\end{equation}\\[-5pt]
where $\tilde{k}(N)=(k(N)-1)/2=\frac{4N^3+3N^2-N}{6}$ denotes the one-side maximal size of TO-ECA with N-sensors and $|\mathbb{Z}|$ is the consecutive segment of $\Phi^u$,
which means $\mathbb{Z}(\mathbb{S},\mathbb{S}',\mathbb{S}'')=[-Z:Z]\subseteq\Phi^u$.
\end{definition}

The redundancy of TO-ECA quantifies the discrepancies between the consecutive segment $\mathbb{Z}$ and the maximal TO-ECA.
If $R_{T} = 1$, all the elements of TO-ECA are distinct with a ULA.
If $R_{T} >1$, either the normalized size $\Phi^u$ is smaller than 1, or there are holes within its $\Phi^u$.

The definition of $R_{T}$ uses two one-sided quantities $(\tilde{k}(N),Z)$ instead of two-sided quantities $(k(N),\mathbb{Z})$.
This definition follows the convention of the (second-order) redundancy in \cite{Moffet1968},
which can be further traced back to the topic of difference basis in number theory \cite{Redei1948,Erdos1948,Leech1956}.
Now discuss the bounds of $R_{T}$. Since $U$ may be $\{0\}$, in which case $R_{T} = \infty$,
we skip the discussion on the lower bound of $R_{T}$.
\begin{theorem}
The redundancy $R_{T}$ of TO-ECA satisfies\\[-5pt]
\begin{equation}
\label{wto12}
\begin{aligned}
R_{T}>L_3(N):&=(1+\frac{2}{3\pi})\frac{\tilde{k}(N)}{\left( \left(\begin{matrix} \small{N} \\ \small{3} \\ \end{matrix} \right) \right)}\\
&=(1+\frac{2}{3\pi})\frac{4N^3+3N^2-N}{N^3+3N^2+2N}.
\end{aligned}
\end{equation}
\end{theorem}

\begin{proof}
Inspired by the proof method in \cite{Redei1948} and \cite{Linel1993},
the lower bound of $R_{T}$ is given by a function of $y$ as follows
\begin{equation}\nonumber
\begin{aligned}
f(y)=\sum_{\substack{0\leq n_i \leq N-1 \\ (i=1,2,3)}}[\cos((l_{n_1}\sx +\sx l_{n_2}\sx +\sx l_{n_3})y) +\cos((l_{n_1}\sx -\sx l_{n_2}\sx -\sx l_{n_3})y)]\\
\end{aligned}
\end{equation}

By the definition of $\mathbb{Z}$, for each integer $z \in [-Z, Z]$, we select one 2-tuple
$T_z= (n_1(z), n_2(z),n_3(z))$ such that $n_1(z) + n_2(z)+n_3(z)= z_1$, $n_1(z) - n_2(z)-n_3(z)= z_2$, $0 \leq n_i(z) \leq N-1 (i=1,2,3)$.
Denoting $\mathcal{T} =\{T_z | z \in [-Z, Z]\}$, we can decompose $f(y)$ into\\[-4pt]
\begin{equation}\nonumber
\begin{aligned}
f(y)&\sx =\sx \sum_{(\sx n_1 , n_2 , n_3\sx ) \in \mathcal{T}}[\cos((l_{n_1}\sx +\sx l_{n_2}\sx +\sx l_{n_3})y)\sx +\sx \cos((l_{n_1}\sx -\sx l_{n_2}\sx -\sx l_{n_3})y)]\\
&\sx +\sx \sum_{(\sx n_1, n_2,n_3\sx ) \notin \mathcal{T}}[\cos((l_{n_1}\sx +\sx l_{n_2}\sx +\sx l_{n_3})y)\sx +\sx \cos((l_{n_1}\sx -\sx l_{n_2}\sx -\sx l_{n_3})y)].
\end{aligned}
\end{equation}

Since $\cos y\leq1$ for all $y\in \mathbb{R}$ and $|\mathcal{T}|=2Z+1$, we have
\begin{equation}\nonumber
\begin{aligned}
f(y)&\leq\sum_{z_1=-Z}^{Z}\cos(z_1y)\sx +\sum_{z_2=-Z}^{Z}\cos(z_2y)
\sx +2\left( \left(
               \begin{matrix}
                 N \\
                 3                    \\
               \end{matrix}
         \right) \right)
             \sx -|\mathcal{T}|.
\end{aligned}
\end{equation}

To further discuss the value of $f(y)$, $\sum\limits_{u=-Z}^{Z}\cos(zy)$ is considered.
Since $\cos(zy)=\frac{e^{izy}+e^{-izy}}{2}$, we can get
\begin{equation}
\label{wto6}
\begin{aligned}
\sum_{z=-Z}^{Z}\cos(zy)=\frac{1}{2}(\sum_{z=-Z}^{Z}e^{izy}+\sum_{z=-Z}^{Z}e^{-izy}),
\end{aligned}
\end{equation}
for the positive of exponent in (\ref{wto6}), it can be rewritten as follows\\[-13pt]
\begin{equation}
\label{wto7}
\begin{aligned}
\Upsilon&\triangleq\sum_{z=-Z}^{Z}e^{izy}=\sum_{z=0}^{Z}e^{izy}+\sum_{z=-Z}^{-1}e^{izy}.
\end{aligned}
\end{equation}

According to the symmetry of $\Upsilon$, we can get\\[-5pt]
\begin{equation}\nonumber
\begin{aligned}
\sum_{z=1}^{Z}e^{-izy}=\sum_{z=-Z}^{-1}e^{izy}.
\end{aligned}
\end{equation}

Therefore, $\Upsilon$ in (\ref{wto7}) can be transformed as\\[-5pt]
\begin{equation}\nonumber
\begin{aligned}
\Upsilon&=1+\sum_{z=1}^{Z}(e^{izy}+e^{-izy})
=1+\sum_{z=1}^{Z}2\cos(zy).
\end{aligned}
\end{equation}

In addition, for $\sum_{z=-Z}^{Z}e^{izy}=e^{-iZy}\sum_{k=0}^{2Z}e^{iky}$,  we can get the following equation from the principle of summing geometric series\\[-8pt]
\begin{equation}
\label{wto8}
\begin{aligned}
\sum_{k=0}^{K}q^k=\frac{1-q^{N+1}}{1-q},\ \ (q=e^{iy}),
\end{aligned}
\end{equation}\\[-6pt]
substituting (\ref{wto8}) into (\ref{wto7}), we can obtain\\[-8pt]
\begin{equation}\nonumber
\begin{aligned}
\sum_{u=-Z}^{Z}e^{izy}=e^{-iZy}\frac{1-e^{i(2Z+1)y}}{1-e^{iy}}.
\end{aligned}
\end{equation}

Since $1-e^{iy}=1-\cos y-i\sin y$, the modulo of it is as follows
\begin{equation}
\label{wto9}
\begin{aligned}
|1-e^{iy}|=2\sin\frac{y}{2}.
\end{aligned}
\end{equation}

Similarly, for $1-e^{i(2Z+1)y}=1-\cos((2Z+1)y)-i\sin((2Z+1)y)$, the modulo of it is as follows
\begin{equation}
\label{wto10}
\begin{aligned}
|1-e^{i(2Z+1)y}|=2\sin[(Z+\frac{1}{2})y].
\end{aligned}
\end{equation}

Therefore, substituting (\ref{wto9}) and (\ref{wto10}) into (\ref{wto6})  we can get
\begin{equation}
\begin{aligned}
\sum_{z=-Z}^{Z}\cos(zy)=\frac{\sin[(Z+\frac{1}{2})y]}{\sin\frac{y}{2}}.
\end{aligned}
\end{equation}

Due to the non negativity of $f(y)$ for all $y$, therefore we can get
\begin{equation}\nonumber
\begin{aligned}
\frac{2\sin[(Z+\frac{1}{2})y]}{\sin\frac{y}{2}}+
2\left( \left(
               \begin{matrix}
                 N \\
                 3                    \\
               \end{matrix}
         \right) \right)
-2Z-1 \geq f(y)\geq0.
\end{aligned}
\end{equation}

By rearranging this inequality and selecting $y=\frac{3\pi}{2Z+1}$ for $Z \geq 1$, we can obtain
\begin{equation}
\label{wto11}
\begin{aligned}
\frac{\left( \left(
               \begin{matrix}
                 N \\
                 3                    \\
               \end{matrix}
         \right) \right)}{Z} \sx \geq\sx (1-\frac{\sin[(Z_2+\frac{1}{2})y]}{Z\sin\frac{y}{2}})+\frac{1}{2Z} \sx >\sx (1+\frac{2}{3\pi}).
\end{aligned}
\end{equation}

The desired inequality can be obtained by multiplying (\ref{wto11}) with $\tilde{k}(N)/ \left( \left(\begin{matrix} N \\ 3 \\ \end{matrix} \right) \right)$ on both sides.

\end{proof}

The lower bound $L_3(N)$ leads to insights into the TO-ECA. Among all sparse arrays, the size of $\mathbb{Z}$ is strictly smaller
than $k(N)$ for $N \geq 2$. The reason is as follows.
Since it can be shown that $L_3(N)$ is an increasing function of $N$,
the redundancy of TO-ECA satisfies $R_{T} > L_3(N) \geq L_3(2) \approx 2.1214$ for $N \geq 2$.
This relation indicates that for $N \geq 2$, $\tilde{k}(N)>Z$, or equivalently $|\mathbb{Z}| < k(N)$.

For sufficiently large $N$, we can derive a upper bound for $|\mathbb{Z}|$,
which is stronger than $k(N)$. Based on (\ref{wto12}), $\tilde{k}(N)/L_3(N)$ is
approximately $0.8488N^3$ in this region. As a result, we have the following asymptotic relation for $|\mathbb{Z}|$

\begin{equation}\nonumber
|\mathbb{Z}|=2Z+1\leq 1+2\frac{\tilde{k}(N)}{L_3(N)}\approx 1.6977 N^3.
\end{equation}\\[-25pt]
\subsection{Redundancy of TO-SDA with three different generators}

\textbf{Case 1:}
For the special case of TO-ECA, the $R_{T}$ of TO-SDA(CNA) is considered.
We can get the $Z_1=(-(N_1^*)^3+(N-\frac{21}{4})(N_1^*)^2+(6N+\frac{19}{2})N_1^*-\frac{17}{4}-5N-1)/2\ (N_1^*=\lfloor\frac{4N+\sqrt{(4N+15)^2+672}-21}{12}\rfloor)$
according to Lemma2.\par
\textbf{Case 2:}
For the special case of TO-ECA, the $R_{T}$ of TO-SDA(SCNA) is considered.
We can get the $Z_2=(-(N_1^*)^3+(N-\frac{21}{4})(N_1^*)^2+(6N+\frac{3}{2})N_1^*+\frac{7}{4}+3N-1) /2\ (N_1^*=\lfloor\frac{4N+\sqrt{(4N+15)^2+288}-21}{12}\rfloor)$
according to Lemma4.\par
\textbf{Case 3:}
For the special case of TO-ECA, the $R_{T}$ of TO-SDA(TNA-II) is considered.
We can get the $Z_3=(-2(N_1^*)^3+(2N+\frac{3}{2})(N_1^*)^2+(\frac{9}{2}-4N)N_1^*+4N^2-\frac{3}{2}N-\frac{47}{8}-1)/2\ (N_1^*=\lceil\frac{3+4N+\sqrt{(4N-9)^2+36}}{12}\rfloor)$
according to Lemma6.

Therefore, the $R_{T}^i(i=1,2,3)$ of the TO-SDA(CNA), TO-SDA(SCNA) and TO-SDA(TNA-II) can be obtained based on the Definition 10 as follows
\begin{equation}\nonumber
\begin{aligned}
&R_{T}^i=\frac{4N^3+3N^2-N}{6Z_i}.
\end{aligned}
\end{equation}

\begin{corollary}
The upper and lower bounds of the redundancy $R_T^i (i=1,2,3)$ with the number of physical sensors varying from 2 to infinity are
\begin{equation}\nonumber
\begin{aligned}
&2.4789\leq R_T^1\leq 9,\ 2.200\leq R_T^2\leq 9,\ 2.1477\leq R_T^3\leq 4.5.
\end{aligned}
\end{equation}
\end{corollary}

\begin{proof}
See Appendix D.
\end{proof}

\section{PERFORMANCE COMPARISON}
In this section, we provide numerical simulations to demonstrate the superior performance of
TO-SDA with three different generators in terms of DOF, coupling leakage, redundancy and the RMSE versus the input SNR, snapshots and the number of sources.
Note that in all DOA estimations, the spatial smoothing MUSIC algorithm \cite{Pal22011}, \cite{Liu2015}, \cite{Piya2012},
\cite{You2021} is used to estimate DOA.
Moreover, we assume that all incident sources have equal power and the number of sources is known.
To evaluate the results quantitatively, the root-mean-square error (RMSE) of the estimation DOAs is defined as an
average over 1000 independent trials:

\begin{equation}\nonumber
\text{RMSE}=\sqrt{\frac{1}{1000D}\sum_{j=1}^{1000}\sum_{i=1}^{D}(\hat{\theta}_i^{j}-\theta_i)^2},
\end{equation}
where $\hat{\theta}_i^{j}$ is the estimate of $\theta_i$ for the $j^{th}$ trial. Similar to \cite{LiuCL2016},
we focus on the DOF obtained by different arrays,
rather than the array aperture, to investigate the overall estimation performance.

\subsection{Comparison of the DOF for Different Arrays}

We compare the DOF of the proposed method with those of TONA \cite{Sharma2023},
FL-NA \cite{Piya2012} and SE-FL-NA \cite{Shen2019} for given the fixed number of physical sensors,
where all other TO-SDAs adopt the array structure for obtaining maximum DOF.
The comparing variations of DOF for six methods are shown in Fig. 4.
It can be seen that the DOF of TO-SDA(TNA-II) are the highest among the six array structures.
Even compared to SE-FL-NA, it is still superior when $N\geq4$.
In addition, the DOF of TO-SDA(SCNA) are more than those of TONA, TO-SDA(CNA) and FL-NA for any number of sensors,
and its DOF are also more than those of SE-FL-NA when $N\leq14$.
\begin{figure}
 \center{\includegraphics[width=6cm]  {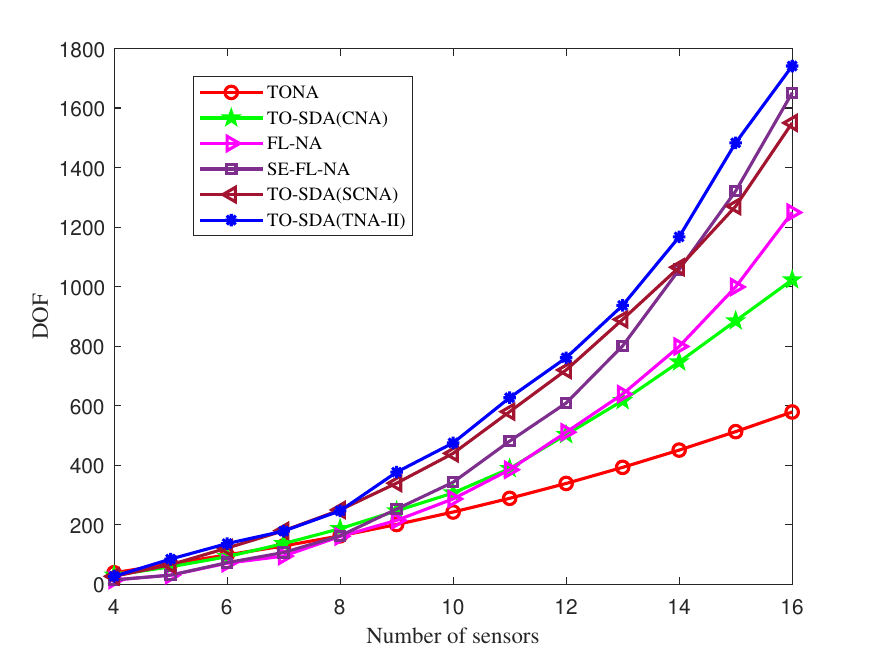}}
 \caption{\label{1} DOF of different arrays}
\end{figure}

\subsection{Redundancy of Different Array Structures}

The redundancy is an important indicator to measure whether the DOF of the current array structure can be further enhanced,
which is compared among TONA and the three TO-SDA designed based on different generators in the simulation.
Redundancies of two arrays with the varying of the number of sensors are shown on the Fig. 5.
It can be seen that the redundancy of TONA is smaller than those of other three TO-SDAs when the number of sensors is less than 8.
While the number of sensors is more than 8, the redundancy of the proposed TO-SDA(TNA-II) is much smaller than those of other three arrays.
\begin{figure}[h]
 \center{\includegraphics[width=6cm]  {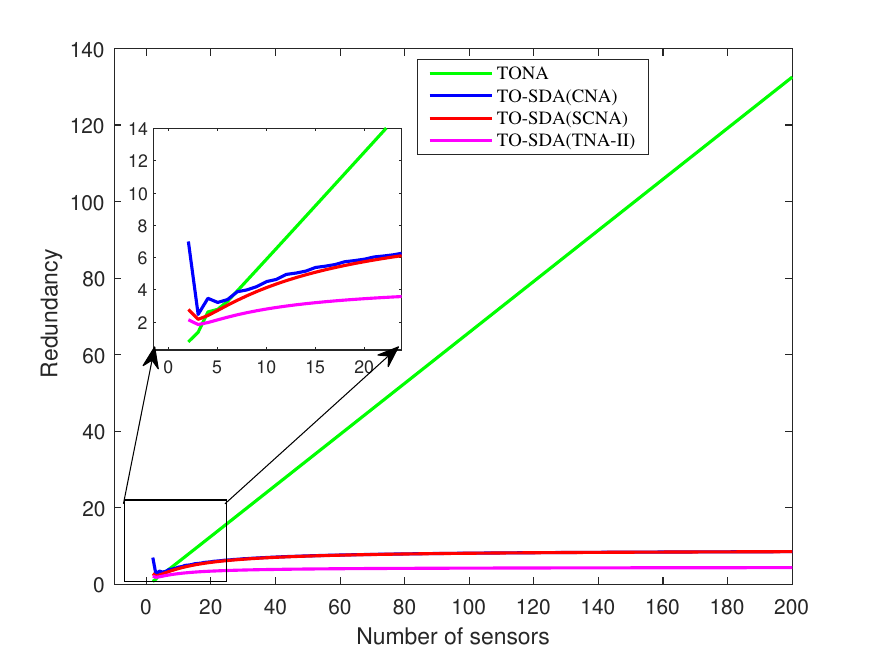}}
 \caption{\label{1} Redundancy of different arrays}
\end{figure}

\subsection{Coupling Leakage}
The mutual coupling performance of the proposed co-array is compared with those of state-of-the-art co-arrays, FL-NA, SE-FL-NA and TONA,
in terms of coupling leakages in this section.
Specifically, the mutual coupling model (\ref{wang5}) is characterized by
$c_1=0.3e^{j\pi/3}$, $B=100$ and $c_l=c_1e^{-j(l-1)\pi/8}/l$, for $2\leq l \leq B$.

Firstly, the coupling leakages are calculated by (\ref{w24}) for different configurations with the number of physical sensors
as 9, 10, 11, 19, 21 and 23, respectively.
The results of the coupling leakages of different structure for six arrays are listed in Table I.
It can be seen that the coupling leakage $L$ of TO-SDA(TNA-II) is smaller than those of other five arrays with any number of sensors,
where the reason is that there exists less physical sensors with unit inter-spacing in TO-SDA(TNA-II)
than those of five arrays, resulting in decreased coupling leakage of TO-SDA(TNA-II).
In addition, the performance of the coupling leakages of TO-SDA(CNA) and TO-SDA(SCNA) are better than that of TONA,
while they are worse than those of FL-NA and SE-FL-NA.
Therefore, the proposed TO-SDA(TNA-II) outperforms other five arrays under any number of physical sensors in terms of coupling leakage.

\begin{table}
\label{tab2}
\begin{center}
\caption{A SUMMARY OF MUTUAL COUPLING LEAKAGE FOR THREE ARRAY STRUCTURES}
\setlength{\tabcolsep}{0pt}
\renewcommand{\arraystretch}{1.5} 
\begin{tabular}{ c  c  c  c  c  c c}
\hline
\hline
\textbf{Array}& \textbf{ \small{FL-NA} }& \textbf{ \small{SE-FL-NA} } & \textbf{\small{TONA}} & \textbf{\small{TO-SDA}} & \textbf{\small{TO-SDA}} & \textbf{\small{TO-SDA}} \\
\textbf{config.} &                  &                 &               & \textbf{\small{(CNA)}} & \textbf{\small{(SCNA)}} & \textbf{\small{(TNA-II)}} \\
\hline
\small{9 Sensors}   & \small{(3,3,3,3)} & \small{(3,3,3,2)} & \small{(6,3)}  & \small{(6,3)} &\small{(6,3)} &\small{(6,3)} \\
$L$          & \small{0.2263}   & \small{0.2257} & \small{0.3395} & \small{0.2477} &\small{0.2161} &\small{\textbf{0.1957}} \\
\hline
\small{10 Sensors}    & \small{(4,3,3,3)} &\small{(3,3,3,3)} & \small{(6,4)}  & \small{(7,3)} &\small{(7,3)} &\small{(6,4)} \\
$L$          & \small{0.2563}     &\small{0.2147} & \small{0.3240} & \small{0.2895} &\small{0.2664} &\small{\textbf{0.1860}}\\
\hline
\small{11 Sensors}   & \small{(4,4,3,3)} & \small{(4,3,3,3)} & \small{(7,4)}   & \small{(6,5)} &\small{(6,5)} & \small{(7,4)} \\
$L$          & \small{0.2477}    &\small{0.2449} & \small{0.3403}  & \small{0.2771} &\small{0.2529} &\small{\textbf{0.1817}} \\
\hline
\small{19 Sensors}   & \small{(6,6,5,5)} &\small{(6,5,5,5)} & \small{(11,8)}  & \small{(10,9)} &\small{(10,9)} &\small{(12,7)}\\
$L$          & \small{0.2460}    &\small{0.2452} & \small{0.3436}  & \small{0.2646} &\small{0.2494} & \small{\textbf{0.2280}}\\
\hline
\small{21 Sensors}   & \small{(6,6,6,6)} &\small{(6,6,6,5)} & \small{(12,9)}  &  \small{(11,10)} &\small{(11,10)} &\small{(14,7)}\\
$L$          & \small{0.2347}    &\small{0.2346} & \small{0.3442}  & \small{0.2554}   &\small{0.2410} &\small{\textbf{0.2393}} \\
\hline
\small{23 Sensors}   & \small{(7,7,6,6)} &\small{(7,6,6,6)} & \small{(15,8)}  &  \small{(12,11)} &\small{(12,11)} &\small{(15,8)}\\
$L$          & \small{0.2464}    &\small{0.2459} & \small{0.3710}  & \small{0.2776}   &\small{0.2657} &\small{\textbf{0.2296}}\\
\hline
\hline
\end{tabular}
\end{center}
\end{table}

Secondly, the visualizations of mutual coupling matrices for different arrays are shown in Fig. 6,
with the number of antennas set as 19,
where the darker blue color corresponds to the smaller value of the non-diagonal elements
in mutual coupling matrices and yellow color represents the values of diagonal elements of matrices.
It is observed in Fig. 6 that the higher mutual coupling of TONA, TO-SDA and FL-NA occurs in the dense sub-array part,
i.e., the nested array, whereas higher mutual coupling of TO-SDA is concentrated at both ends of the
$\mathcal{G}$, where the senor spacing is denser.

\begin{figure}
  \centering
  \subfigure[]{
    \label{fig:subfig:threefunction}
    \includegraphics[scale=0.18]{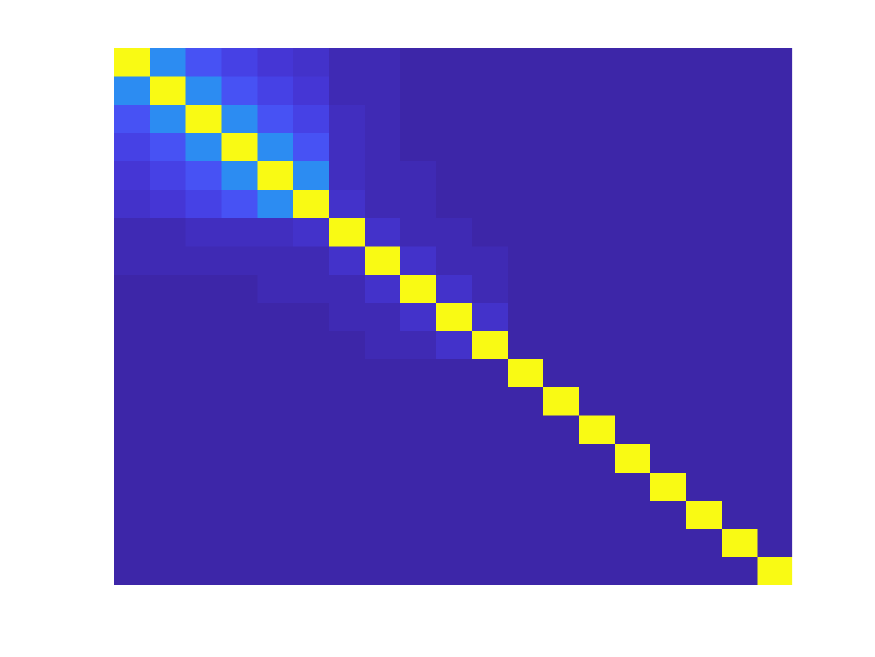}}
    \hspace{0in} 
     \subfigure[]{
    \label{fig:subfig:threefunction}
    \includegraphics[scale=0.18]{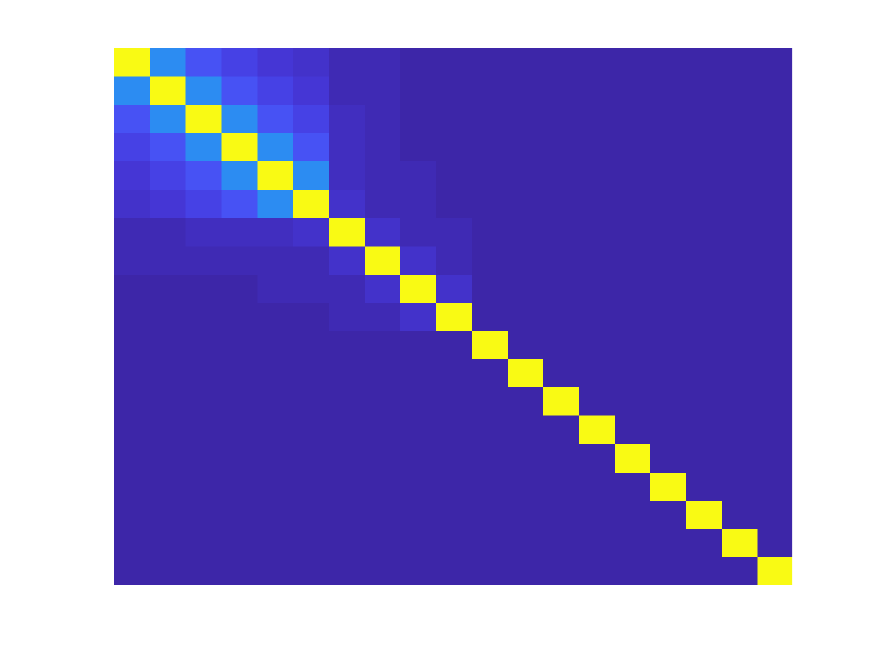}}
    \hspace{0in} 
     \subfigure[]{
    \label{fig:subfig:onefunction}
    \includegraphics[scale=0.18]{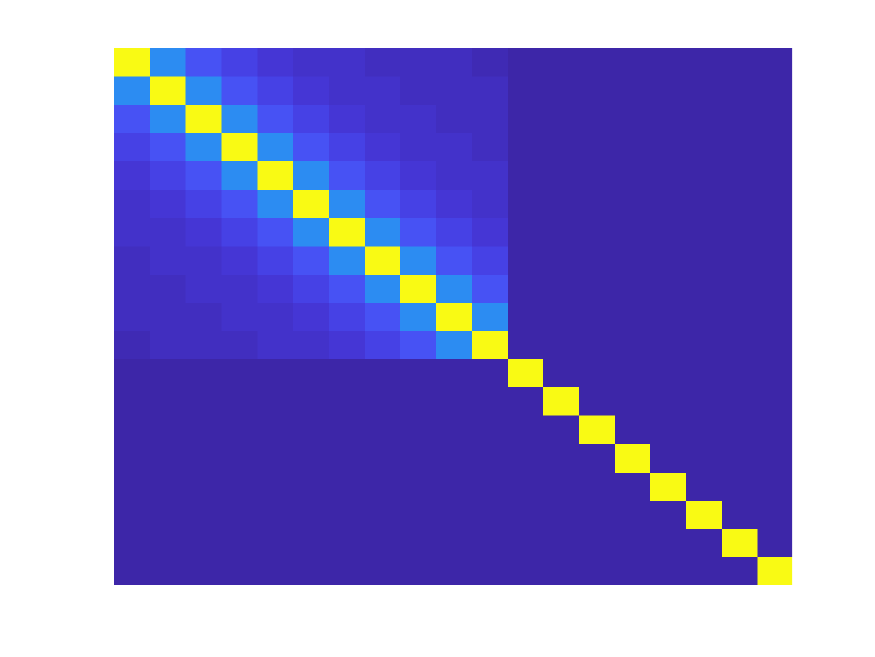}}
  \hspace{0in} 
  \subfigure[]{
    \label{fig:subfig:threefunction}
    \includegraphics[scale=0.18]{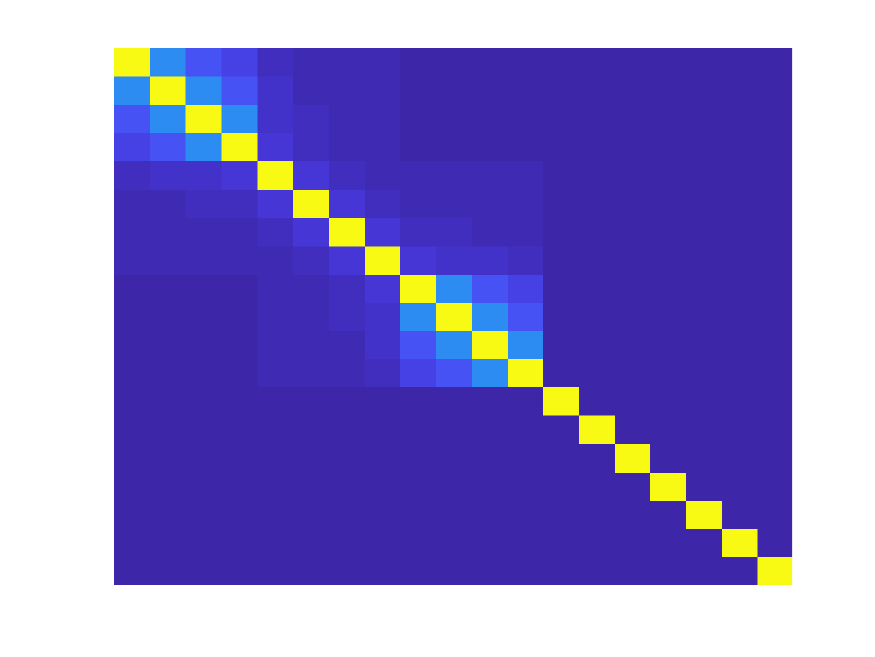}}
  \hspace{0in} 
  \subfigure[]{
    \label{fig:subfig:threefunction}
    \includegraphics[scale=0.18]{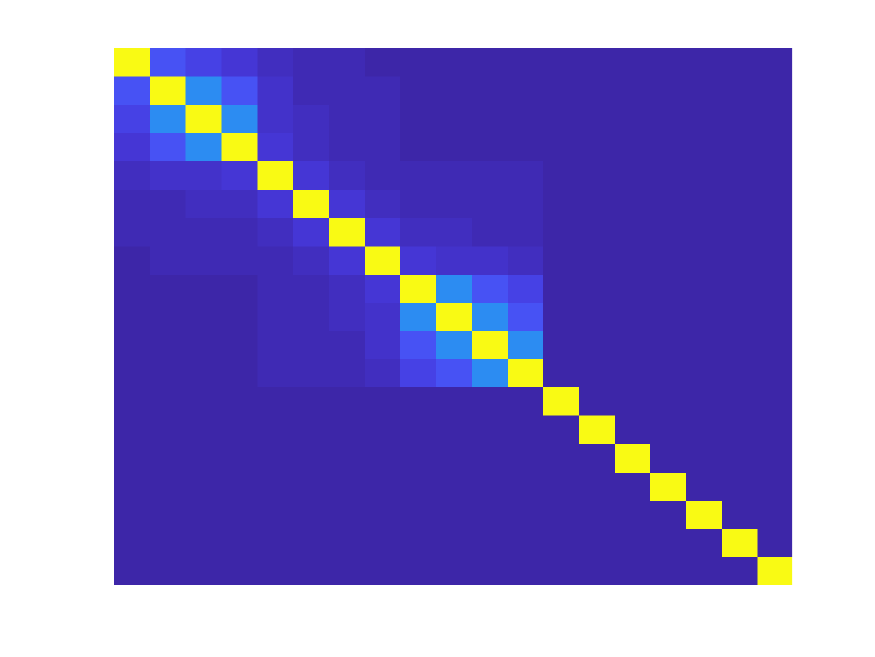}}
  \hspace{0in} 
    \subfigure[]{
    \label{fig:subfig:threefunction}
    \includegraphics[scale=0.18]{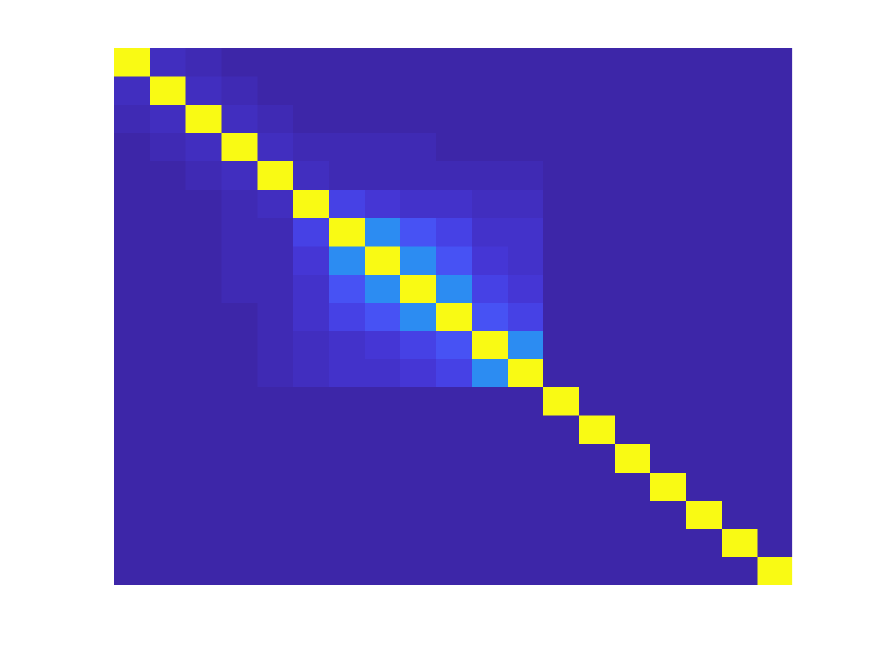}}
    \hspace{0in} 

  \caption{The magnitudes of the mutual coupling matrices of five arrays with 19-sensors. (a) FL-NA. (b) SE-FL-NA. (c) TONA. (d) TO-SDA(CNA). (e) TO-SDA(SCNA). (f) TO-SDA(TNA-II).}
\end{figure}

\subsection{Resolution of Different Array Structures}
\begin{figure}[h]
  \centering
  \subfigure[]{
    \label{fig:subfig:onefunction}
    \includegraphics[scale=0.146]{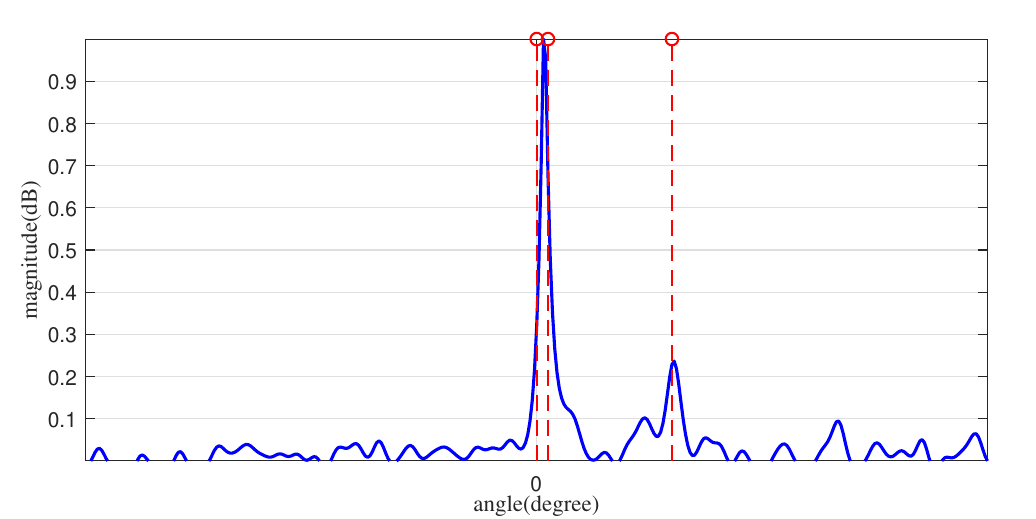}}
  \hspace{0in} 
  \subfigure[]{
    \label{fig:subfig:onefunction}
    \includegraphics[scale=0.146]{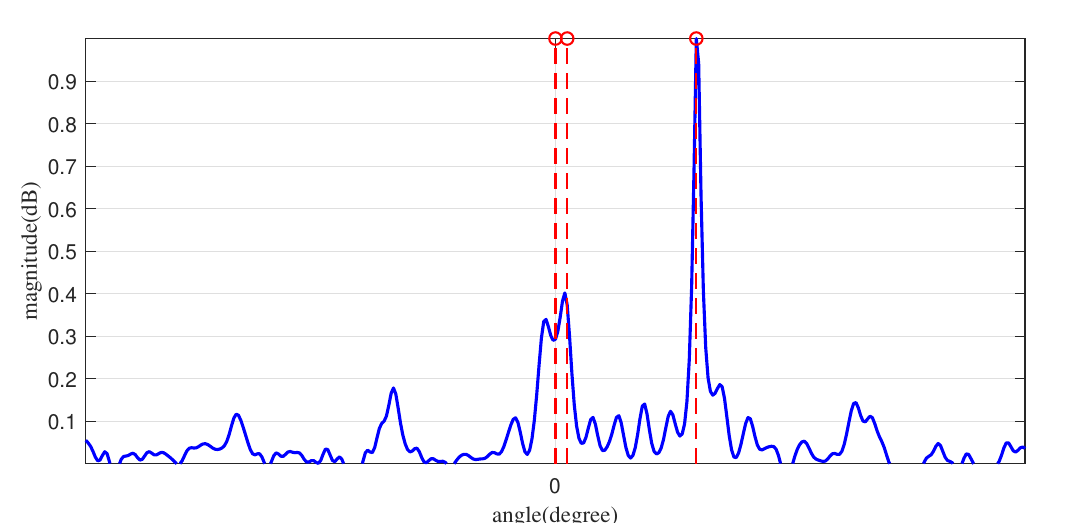}}
  \hspace{0in} 
  \subfigure[]{
    \label{fig:subfig:onefunction}
    \includegraphics[scale=0.146]{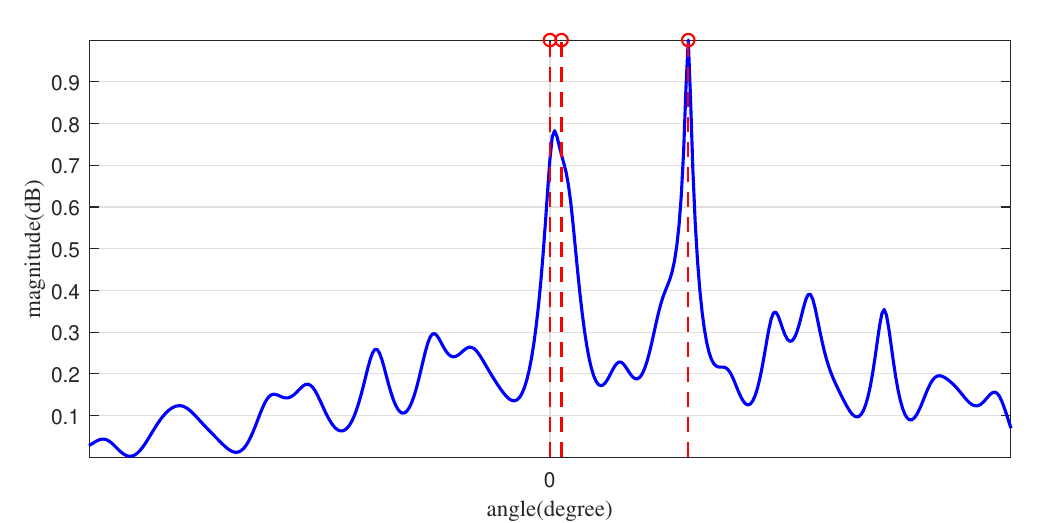}}
  \hspace{0in} 
  \subfigure[]{
    \label{fig:subfig:threefunction}
    \includegraphics[scale=0.15]{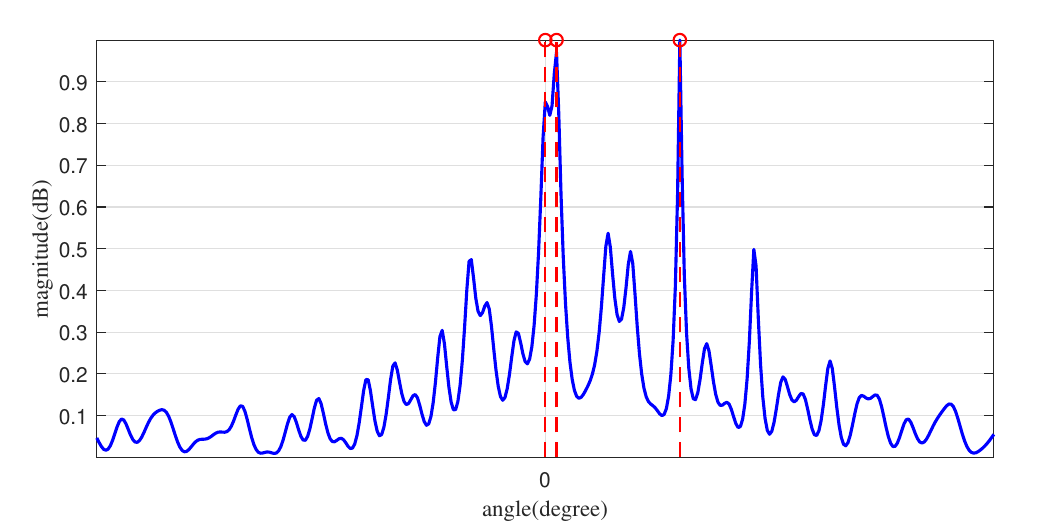}}
  \subfigure[]{
    \label{fig:subfig:threefunction}
    \includegraphics[scale=0.15]{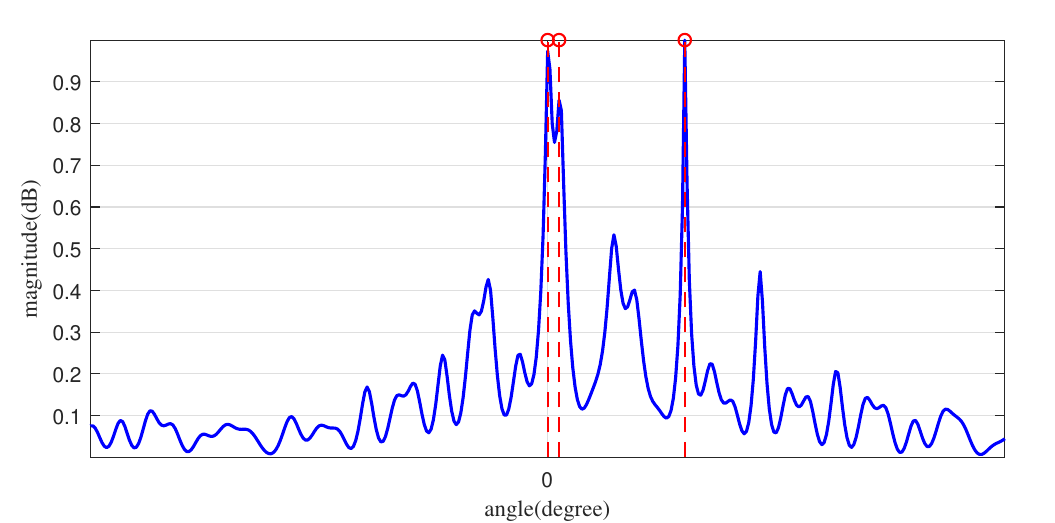}}
  \subfigure[]{
    \label{fig:subfig:threefunction}
    \includegraphics[scale=0.15]{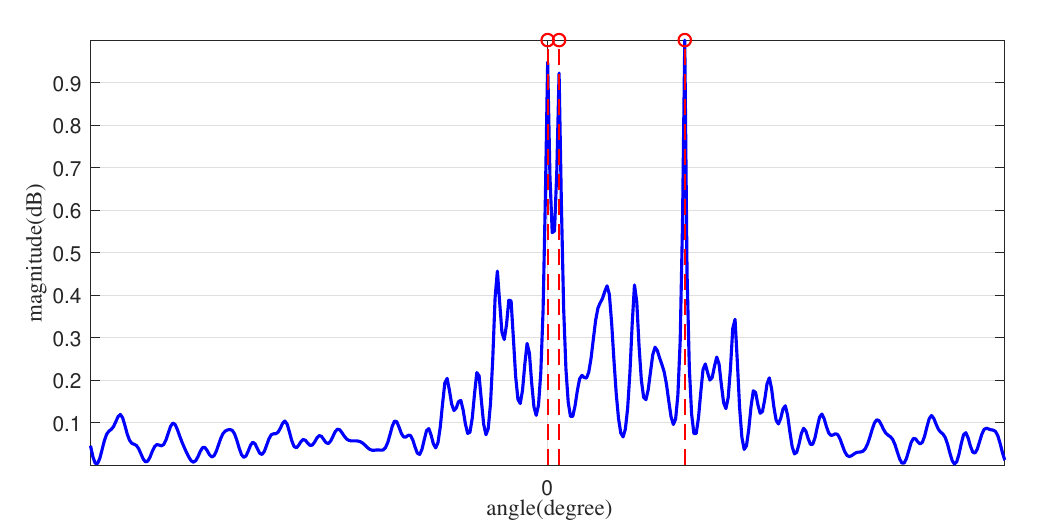}}
  \caption{The special case of DOA estimation result for five arrays with 9-sensors when sources are located at $0^{\circ}$, $0.5^{\circ}$ and $6^{\circ}$.
$SNR = 0$ dB and $K = 10000$. (a) FL-NA. (b) SE-FL-NA. (c) TONA. (d) TO-SDA(CNA). (e) TO-SDA(SCNA). (f) TO-SDA(TNA-II) }
\end{figure}

The resolution is an important metric of the performance for DOA estimation,
wherein high resolution representation enables precise detection in medical imaging.
Further, the resolution is compared among FL-NA, SE-FL-NA, TONA, TO-SDA(CNA), TO-SDA(SCNA)
and TO-SDA(TNA-II) in the simulation. 9 physical sensors are used to construct six arrays,
while the angle of one source fixed at $0^{\circ}$, and the angle of other sources vary from 0 to 10.
The SNR and snapshots are set as 0 dB and 10000, respectively.
The angles of 3 uncorrelated sources are $0^{\circ}$, $0.5^{\circ}$ and $6^{\circ}$ respectively in Fig 7.
It can be seen that TO-SDA(SCNA) and TO-SDA(TNA-II) both can distinguish two sources with the difference of incidence angles reduced to $0.5^{\circ}$, while others can not.
However, TO-SDA(TNA-II) exhibits lower valley than that of TO-SDA(SCNA).

\subsection{DOA Estimation Based on Third-Order Cumulants Without Mutual Coupling}
We compare the DOA estimation performance without mutual coupling versus input SNR, snapshots and the number of sources for the three TO-SDAs to
those of TONA, FL-NA and SE-FL-NA in this part,
where 9 physical sensors are used to construct six arrays.

%
%

The RMSE versus the input SNR, snapshots and the number of sources are studied in the following numerical simulations.
In the first numerical simulation, there are 12 uncorrelated sources uniformly located at $-60^{\circ}$ to $60^{\circ}$
and the snapshots setting as 12000. The SNR ranges from -10dB to 10dB with an interval of 2dB.
The results of RMSE versus SNR for different arrays are shown in Fig. 8(a),
where it can be seen that as the SNR increases, the RMSEs of all arrays decrease,
however the RMSE of TO-SDA(TNA-II) remaining the lowest.
The second numerical simulation studies the DOA estimation performance with respect to the snapshots changing from 8000 to 18000
and the SNR setting as 2dB.
The results of RMSE versus snapshots are shown in Fig. 8(b), and a similar conclusion can be obtained that
the RMSE of TO-SDA(TNA-II) is lower than those of other five arrays.
In the third numerical simulation, the number of sources change from 23 to 28.
The results of RMSE versus the number of sources are shown in Fig. 8(c),
where it can be seen that as the number of sensors increases, the RMSE of TO-SDA(TNA-II) remains the smallest compared to other five arrays, indicating its superior performance.
\begin{figure}
  \centering
  \subfigure[ RMSE versus SNR]{
    \label{fig:subfig:onefunction}
    \includegraphics[scale=0.18]{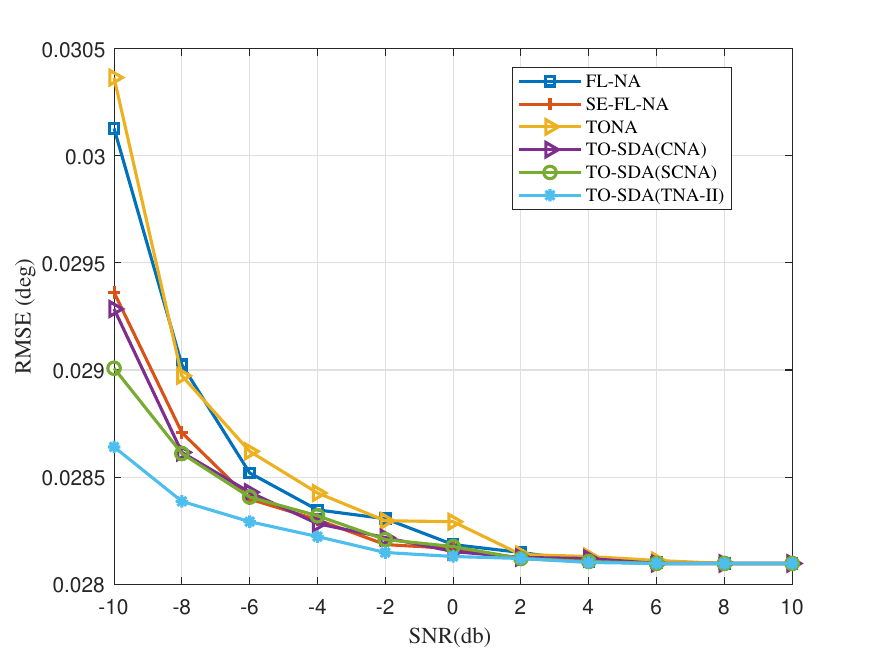}}
   \hspace{0in} 
  \subfigure[RMSE versus Snapshots]{
    \label{fig:subfig:threefunction}
    \includegraphics[scale=0.18]{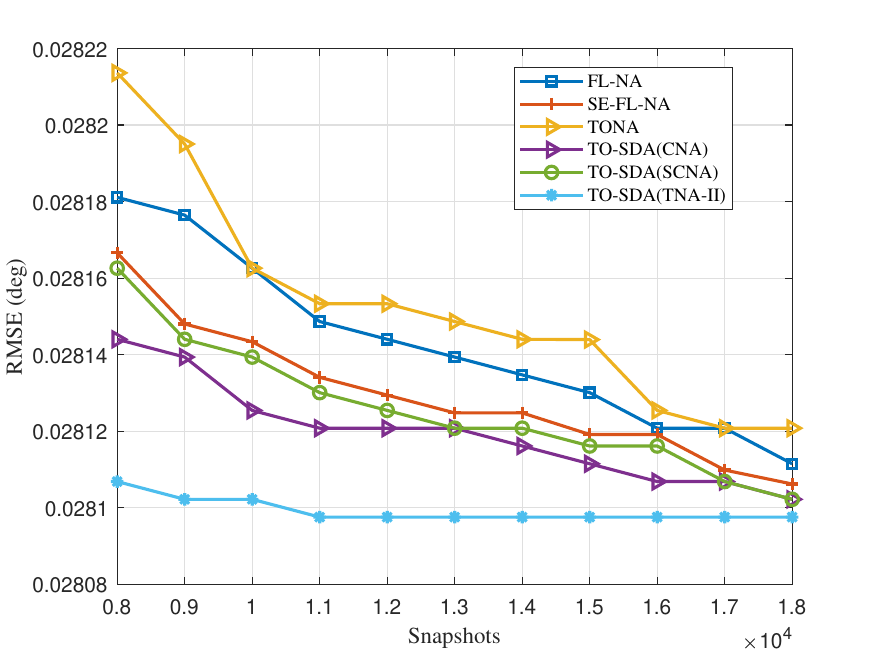}}
    \hspace{0in} 
  \subfigure[RMSE versus Number of Sources]{
    \label{fig:subfig:threefunction}
    \includegraphics[scale=0.18]{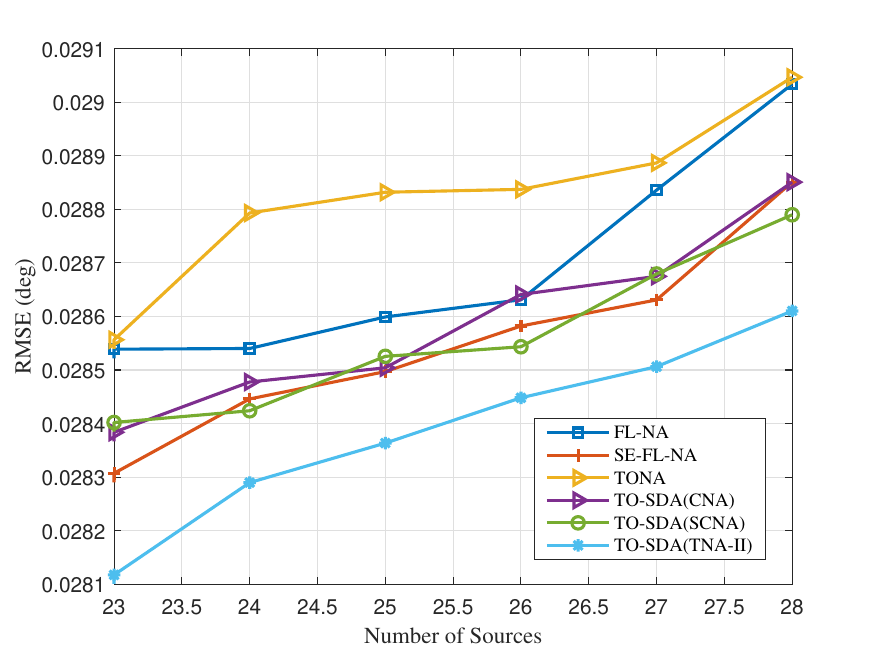}}
  \caption{DOA estimation performance without mutual coupling based on third-order cumulants}
\end{figure}

\subsection{DOA Estimation Based on Third-Order Cumulants With Mutual Coupling}
We compare the DOA estimation performance with mutual coupling versus input SNR, snapshots and the number of sources for the three TO-SDCAs to
those of TONA, FL-NA and SE-FL-NA in this part,
where 9 physical sensors are used to construct six arrays.

The RMSE versus the input SNR, snapshots and the number of sources are studied in the following numerical simulations.
In the first numerical simulation, there are 12 uncorrelated sources uniformly located at $-60^{\circ}$ to $60^{\circ}$
and the snapshots setting as 12000. The SNR ranges from -10dB to 10dB with an interval of 2dB.
The results of RMSE versus SNR for different arrays are shown in Fig. 9(a),
where it can be seen that as the SNR increases, the RMSEs of all arrays decrease,
however the RMSE of TO-SDA(TNA-II) remaining the lowest.
The second numerical simulation studies the DOA estimation performance with respect to the snapshots changing from 8000 to 18000
and the SNR setting as 2dB.
The results of RMSE versus snapshots are shown in Fig. 9(b), and a similar conclusion can be obtained that
the RMSE of TO-SDA(TNA-II) is lower than those of TONA and FL-NA.
In the third numerical simulation, the number of sources change from 23 to 28.
The results of RMSE versus the number of sources are shown in Fig. 9(c),
where it can be seen that as the number of sensors increases, the RMSE of SE-FL-NA, TO-SDA(CNA) and TO-SDA(SCNA) increase steeply
than those of other three arrays.
Notably, the RMSE of TO-SDA(TNA-II) remains the smallest compared to other five arrays, indicating its superior performance with respect to the coupling effects.
\begin{figure}
  \centering
  \subfigure[ RMSE versus SNR]{
    \label{fig:subfig:onefunction}
    \includegraphics[scale=0.18]{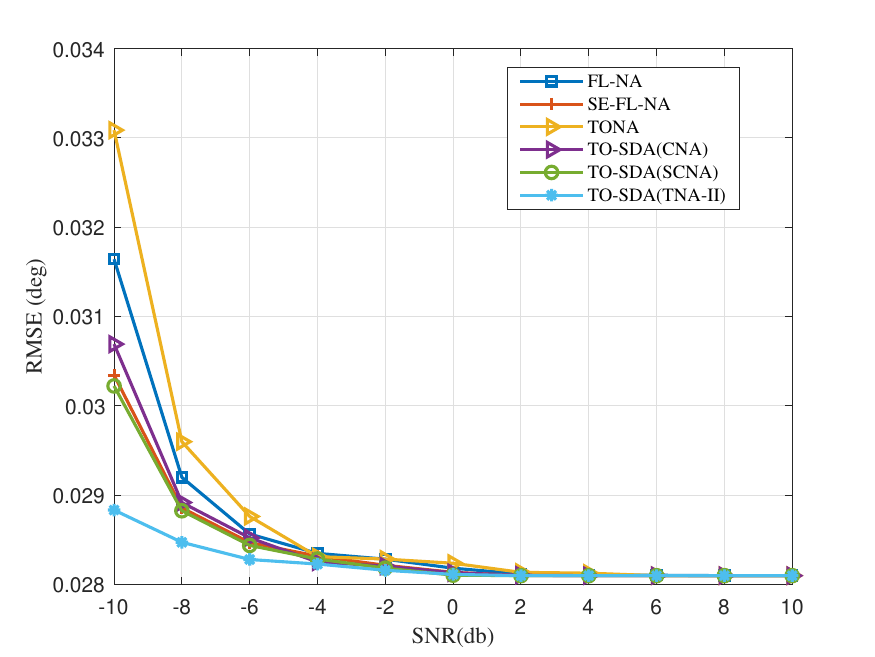}}
   \hspace{0in} 
  \subfigure[RMSE versus Snapshots]{
    \label{fig:subfig:threefunction}
    \includegraphics[scale=0.18]{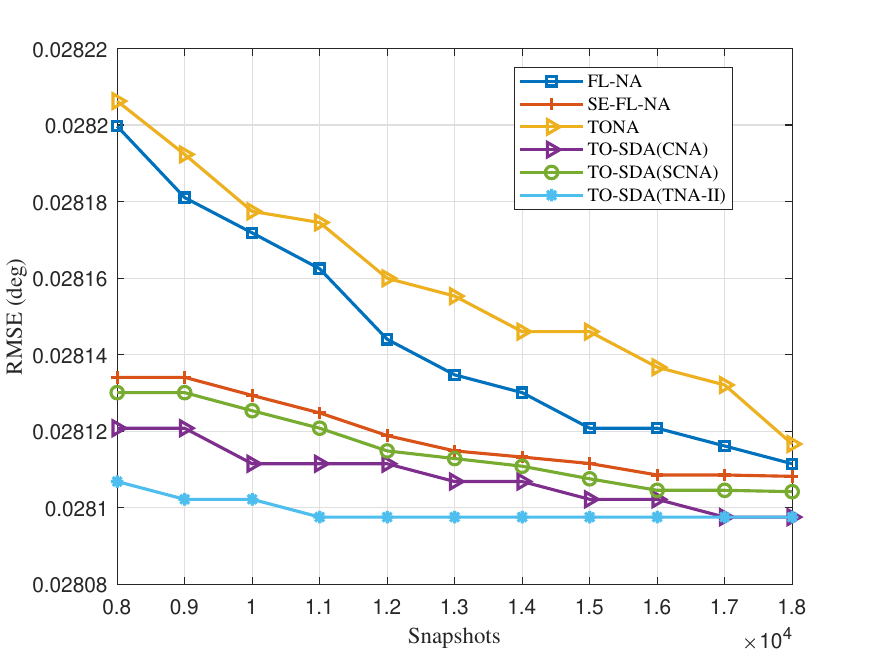}}
    \hspace{0in} 
  \subfigure[RMSE versus Number of Sources]{
    \label{fig:subfig:threefunction}
    \includegraphics[scale=0.18]{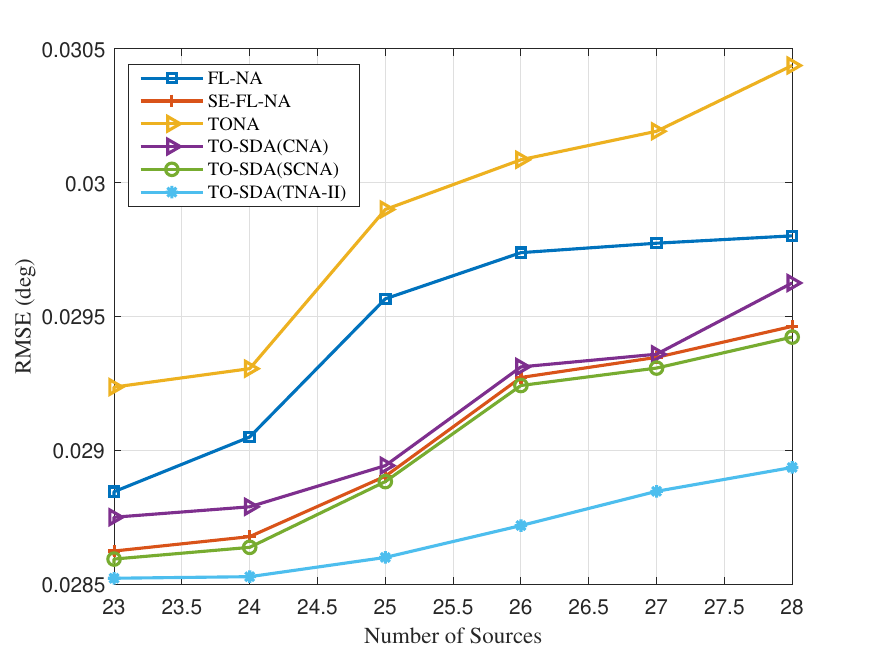}}
  \caption{DOA estimation performance with mutual coupling based on third-order cumulants}
\end{figure}

\section{Conclusion}
The paper devise a novel method to design a sparse sensor arrays, namely GTOA with any generator based on the TO-ECA,
and the redundancy of TO-ECA is defined.
Further, three novel SLA with different fixed generators, namely TO-SDA(CNA), TO-SDA(SCNA) and TO-SDA(TNA-II) are proposed
and the redundancy of the three TO-SDAs are defined.
Specifically, the three TO-SDAs all consist of two sub-arrays with a given number of sensors.
Sub-array 1 and 2 are capable to produce a hole-free sum and difference co-arrays with more consecutive lags.
Based on the criterion, sub-array 1 is selected as CNA, SCNA and TNA-II,
and 2 is a sparse linear array with large inter-spacing between two physical sensors.
This arrangement can yield closed-form expressions for the antenna positions of the proposed three TO-SDAs.
Therefore, based on this design, the proposed three TO-SDAs offers significantly larger DOF than those of other existing arrays.
Especially, the DOF of TO-SDA(TNA-II) are even more than those of SE-FL-NA, and its redundancy and coupling leakage are less than those of SE-FL-NA.
The enhanced DOF increase the number of resolvable sources with the given number of physical sensors.
Meanwhile, compared to existing TONA, FL-NA and SE-FL-NA, TO-SDA(TNA-II) exhibits lower coupling effects.
In a short, the novel three sparse sensor arrays note only can enhance the DOF and reduce redundancy with the same number of sensors to reduce hardware costs,
but also can increase the array aperture to realize higher resolution,
and reduce the mutual coupling effects to decrease the mutual interference between sensors,
and improve the DOA estimation performance.

\section*{Acknowledgment}
This work was supported by the National Natural Science Foundation of China (Grant No. 62371400).\\

{\appendices

\section{The Proof of Lemma 2}
It follows from Eq. (\ref{wh1}) that the DOF of the TO-SDA is
\begin{equation}\nonumber
\text{DOF}=(6+8N_2)[(M_1-1)+M_2(M_1+1)]+2N_2+1.
\end{equation}

Furthermore, the total number of antennas is
\begin{equation}\nonumber
N=2M_1+M_2+N_2.
\end{equation}

Therefore, to obtain the maximum DOF of TO-SDA, it is equivalent to solve the following optimization problem
\begin{equation}\nonumber
\begin{aligned}
&\underset{M_1,M_2,N_2\in \mathbb{N}_+}{maximize} (6+8N_2)[(M_1-1)+M_2(M_1+1)]+2N_2+1,\\
&subject\ to\ 2M_1+M_2+N_2=N. \ \ \ \ \ \ \ \ \ \ \ \ \ \ \ \ \ \ \ \ \ \ \ \ \ (P1)
\end{aligned}
\end{equation}

Firstly, assuming $N_1$ for $\mathcal{G}$ is known, to obtain the maximum DOF $\mathcal{G}$ based on SCA,
variables $M_1$ and $M_2$ are set as \cite{Robin2017}
\begin{equation}\nonumber
M_1=\lceil\frac{N_1-1}{4}\rfloor, M_2=N_1-2\lceil\frac{N_1-1}{4}\rfloor.
\end{equation}

Further, substituting $M_1$, $M_2$ and $N_2=N-N_1$ into objective function, we can get
\begin{equation}
\label{ap1}
\begin{aligned}
\text{DOF}\sx &=\sx (6+8(N-N_1))[(\lceil\frac{N_1-1}{4}\rfloor-1)+(N_1-2\lceil\frac{N_1-1}{4}\rfloor)\\
&\ \ \ \ \times(\lceil\frac{N_1-1}{4}\rfloor+1)]+2(N-N_1)+1.
\end{aligned}
\end{equation}

Although (P1) is an problem with no general closed-form solution,
the maximum possible DOF maybe found under the relaxation that $M_1,M_2,N_2\in \mathbb{R}_+$.
At this time, (\ref{ap1}) can be changed as follows form
\begin{equation}\nonumber
\begin{aligned}
f_1(N_1)&\triangleq\text{DOF}\\
&=-N_1^3+(N-\frac{21}{4})N_1^2+(6N+\frac{19}{2})N_1-\frac{17}{4}-5N.
\end{aligned}
\end{equation}

For the function $f_1(N_1)$, when the total number of physical sensors $N$ is given, the unknown parameter is only $N_1$,
which is a cubic function $f_1(N_1)$ about $N_1$.
There may be a local maximum for $f_1(N_1)$ within the independent variable interval $[0, N]$.
To discuss the local maximum of $f_1(N_1)$ and obtain the stationary point of $f_1(N_1)$,
the first derivative of $f_1(N_1)$ is derived firstly as follows
\begin{equation}\nonumber
\frac{\partial f_1(N_1)}{\partial N_1}=\frac{19}{2}-\frac{21}{2}N_1+6N+2NN_1-3N_1^2.
\end{equation}

Therefore, the stationary points of $f_1(N_1)$ can be derived as follows when $\frac{\partial f_1(N_1)}{\partial N_1}=0$
\begin{equation}\nonumber
\begin{aligned}
&N_{11},N_{12}=\frac{4N\pm\sqrt{(4N+15)^2+672}-21}{12}.
\end{aligned}
\end{equation}

Since $N_{11}$ and $N_{12}$ represent the number of sensors, that means they are greater than 0.
When $N_{12}=\frac{4N-\sqrt{(4N+15)^2+672}-21}{12}<0$ is not satisfied the condition,
thus only $N_{11}$ is considered for the local maximum of $f_1(N_1)$.

Furthermore, to further derive the local maximum of $f_1(N_1)$,
we compare the results of $f_1(N_1)$ at the endpoints of interval $[0,N]$ and stationary point $N_{11}$, which are calculated as follows
\begin{equation}\nonumber
\begin{aligned}
&f_1(0)=-\frac{17}{4}-5N<0,\
f_1(N)=\frac{3}{4}N^2+\frac{N}{2}-\frac{17}{4},\\
&f_1(N_{11})=\frac{2N^3}{27}+(\frac{N^2}{54}+\frac{5N}{36}+\frac{299}{288})\sqrt{(4N+15)^2+672}\\
&\ \ \ \ +\frac{5N^2}{6}-\frac{149N}{24}-\frac{1011}{32}.
\end{aligned}
\end{equation}

Since $f_1(N_1)$ represents DOF, it must be greater than 0 that the case of $f(0)$ is not considered.
Further, to obtain the local maximum, the relationship between the results of $f_1(N)$ and $f_1(N_{11})$ is compared as follows.
To lighten the notations $\Lambda(N)\triangleq16N^2 + 120N + 897$
\begin{equation}\nonumber
\begin{aligned}
&\Delta(N)=f_1(N_{11})-f_1(N)\\
&=-\frac{257N}{24}+\frac{N^2}{12}-\frac{875}{32}+ (\frac{5N}{36}+\frac{N^2}{54}+\frac{299}{288})\sqrt{(\Lambda(N))},\\
&\frac{\partial \Delta(N)}{\partial N}=\frac{2N^2}{9}+\frac{N}{6}-\frac{257}{24}+(\frac{5}{36}+\frac{N}{27})\sqrt{(\Lambda(N))}\\
&\ \ \ \ +(\frac{5N}{72}+\frac{N^2}{108}+\frac{299}{576})\frac{32N+120}{\sqrt{(\Lambda(N))}}.
\end{aligned}
\end{equation}

Since $N\geq2$, we can get the follow inequality
\begin{equation}\nonumber
\begin{aligned}
\frac{\partial \Delta(N)}{\partial N}\geq \frac{236}{891}N^2+\frac{331}{198}N-\frac{4601}{1188}.
\end{aligned}
\end{equation}

When $\frac{\partial \Delta(N)}{\partial N}=0$, the solutions are $N_{21}\approx1.8021$, $N_{22}\approx-8.1136$.
Since $N\geq2$, only $N_{21}\approx1.8021$ is considered.
According to the properties of quadratic functions, it can be inferred that $\frac{\partial \Delta(N)}{\partial N}\geq 0$ when $N\geq N_{21}$.
At this time, $\Delta(N)$ is a monotonically increasing function when $N\geq 2$.
In addition, $\Delta(1)\approx 0.5323>0$, $\Delta(2)\approx 0.3382>0$,
therefore $\Delta(N)>0$ for all $N\in \mathbb{N}^+$, which means $f_1(N_{11})>f_1(N)$.
Thus, the local maximum solution of $f_1(N_1)$ for $N_1\in[0, N]$
is obtained at the stationary point $N_1^*\triangleq N_{11}=\frac{4N+\sqrt{(4N+15)^2+672}-21}{12}$.
Further, $N_1^*=\lceil\frac{4N+\sqrt{(4N+15)^2+672}-21}{12}\rfloor$ due to $N_1^*$ representing the number of sensors.
And the maximum DOF is $-(N_1^*)^3+(N-\frac{21}{4})(N_1^*)^2+(6N+\frac{19}{2})N_1^*-\frac{17}{4}-5N$.

\section{The Proof of Lemma 4}
It follows from Eq. (\ref{wh1}) that the DOF of the TO-SDA(SCNA) is
\begin{equation}\nonumber
\text{DOF}=(6+8N_2)[M_1+M_2(M_1+1)]+2N_2+1.
\end{equation}

Furthermore, the total number of sensors is
\begin{equation}\nonumber
N=2M_1+M_2+N_2.
\end{equation}

Therefore, to obtain the maximum DOF of TO-SDA(SCNA), it is equivalent to solve the following optimization problem
\begin{equation}\nonumber
\begin{aligned}
&\underset{M_1,M_2,N_2\in \mathbb{N}_+}{maximize} (6+8N_2)[M_1+M_2(M_1+1)]+2N_2+1,\\
&subject\ to\ 2M_1+M_2+N_2=N. \ \ \ \ \ \ \ \ \ \ \ \ \ \ \ \ \ \ \ \ \ \ \ \ \ (P2)
\end{aligned}
\end{equation}

Firstly, assuming $N_1$ for $\mathcal{G}$ is known, to obtain the maximum DOF $\mathcal{G}$ based on SCA,
variables $M_1$ and $M_2$ are set as \cite{Robin2017}
\begin{equation}\nonumber
M_1=\lceil\frac{N_1-1}{4}\rfloor, M_2=N_1-2\lceil\frac{N_1-1}{4}\rfloor.
\end{equation}

Further, substituting $M_1$, $M_2$ and $N_2=N-N_1$ into objective function, we can get
\begin{equation}
\label{ap3}
\begin{aligned}
\text{DOF}&=(6+8(N-N_1))[(\lceil\frac{N_1-1}{4}\rfloor-1)+(N_1-2\lceil\frac{N_1-1}{4}\rfloor)\\
&\ \ \ \ \times(\lceil\frac{N_1-1}{4}\rfloor+1)]+2(N-N_1)+1.
\end{aligned}
\end{equation}

Although (P2) is an problem with no general closed-form solution,
the maximum possible DOF maybe found under the relaxation that $M_1,M_2,N_2\in \mathbb{R}_+$.
At this time, (\ref{ap3}) can be changed as follows form
\begin{equation}\nonumber
\begin{aligned}
f_2(N_1)&\sx \triangleq\sx \text{DOF}\sx =\sx -N_1^3\sx +\sx (N\sx -\sx \frac{21}{4})N_1^2\sx +\sx (6N\sx +\sx \frac{3}{2})N_1\sx +\sx \frac{7}{4}\sx +\sx 3N.
\end{aligned}
\end{equation}

For the function $f_2(N_1)$, its local maximum solution in $N_1\in[0, N]$ is obtained
at the stationary point $N_1^*\triangleq N_{11}=\frac{4N+\sqrt{(4N+15)^2+288}-21}{12}$ by using the same method as $f_1(N_1)$.
Further, $N_1^*=\lceil\frac{4N+\sqrt{(4N+15)^2+288}-21}{12}\rfloor$ due to $N_1^*$ representing the number of sensors.
And the maximum DOF is $-(N_1^*)^3+(N-\frac{21}{4})(N_1^*)^2+(6N+\frac{3}{2})N_1^*+\frac{7}{4}+3N$.

\section{The Proof of Lemma 6}
It follows from Eq. (\ref{to9}) that the DOF of the TO-SDA(TNA-II) is discussed as the following two cases.

\textbf{case 1: $\mathbf{9\leq N\leq14}$.}

When $9\leq N\leq14$, $5\leq N_1\leq9$ can be obtained.
Further, the specific structure of $\mathcal{G}$ can be got.
Additionally, the longest consecutive segments of the second-order SCA and third-order SCA are $2(M_1(M_2+1)+J)-2$ and $3(M_1(M_2+1)+J)-2$.
Therefore, $\lambda_1=2(M_1(M_2+1)+J)-2$ and $\lambda_2=3(M_1(M_2+1)+J)-2$, and the DOF of TO-SDA(TNA-II) can be calculated as
\begin{equation}\nonumber
\text{DOF}=2[(5+4N_2)(M_1(M_2+1)+J)]-6N_2-5.
\end{equation}

\textbf{case 2: $\mathbf{N\leq8}$ and $\mathbf{N\geq15}$.}

When $N\leq8$ and $N\geq15$, $N_1\leq5$ and $N_1\geq10$ can be obtained.
Further, the specific structure of $\mathcal{G}$ can be got.
Additionally, the longest consecutive segments of the second-order SCA and third-order SCA are $2(M_1(M_2+1)+J)$ and $3(M_1(M_2+1)+J)$.
Therefore, $\lambda_1=2(M_1(M_2+1)+J)$ and $\lambda_2=3(M_1(M_2+1)+J)$, and the DOF of TO-SDA(TNA-II) can be calculated as
\begin{equation}\nonumber
\begin{aligned}
\text{DOF}&=2[(4N_2+1)(M_1(M_2+1)+J)+(N_2-1)\\
&+4(M_1(M_2+1))+4J]+1.
\end{aligned}
\end{equation}\\[-7pt]
Furthermore, according to \cite{WangY2024}, $J=\lceil \frac{N_1}{2} \rceil-1$, and the total number of sensors is\\[-8pt]
\begin{equation}\nonumber
N=M_1+M_2+N_2.
\end{equation}\\[-14pt]
Therefore, to obtain the maximum DOF of TO-SDA(TNA-II), it is equivalent to solve the following optimization problem
\begin{equation}\nonumber
\begin{aligned}
&\underset{M_1,M_2,N_2\in \mathbb{N}_+}{maximize} 2[(4N_2+1)(M_1(M_2+1)+J)+(N_2-1)\\
&\ \ \ \ \ \ \ \ \ \ \ \ \ \ \ \ +4(M_1(M_2+1))+4J]+1,\\
&subject\ to\ M_1+M_2+N_2=N. \ \ \ \ \ \ \ \ \ \ \ \ \ \ \ \ \ \ \ \ \ \ \ \ \ (P3)
\end{aligned}
\end{equation}

Firstly, assuming $N_1$ for $\mathcal{G}$ is known, to obtain the maximum DOF $\mathcal{G}$ based on SCA,
variables $M_1$ and $M_2$ are set as \cite{WangY2024}
\begin{equation}\nonumber
M_1=N_1-\lceil\frac{2N_1-1}{4}\rfloor, M_2=\lceil\frac{2N_1-1}{4}\rfloor.
\end{equation}

Further, substituting $M_1$, $M_2$ and $N_2=N-N_1$ into objective function, we can get
\begin{equation}
\label{to10}
\begin{aligned}
\text{DOF}&=2[(4(N \sx - \sx N_1)\sx +\sx 1)((N_1\sx -\sx \lceil\frac{2N_1\sx -\sx 1}{4}\rfloor)(\lceil\frac{2N_1\sx -\sx 1}{4}\rfloor\sx +\sx 1)\\
&\sx +\sx \lceil \frac{N_1}{2} \rceil\sx -\sx 1)+(N\sx -\sx N_1\sx -\sx 1)\sx +\sx 4((N_1\sx -\sx \lceil\frac{2N_1\sx -\sx 1}{4}\rfloor)\\
&\times(\lceil\frac{2N_1\sx -\sx 1}{4}\rfloor\sx +\sx 1))\sx +\sx 4(\lceil \frac{N_1}{2} \rceil\sx -\sx 1)]\sx +\sx 1,\\
\end{aligned}
\end{equation}

Although (P3) is an problem with no general closed-form solution,
the maximum possible DOF maybe found under the relaxation that $M_1,M_2,N_2\in \mathbb{R}_+$.
At this time, (\ref{to10}) can be changed as follows form
\begin{equation}\nonumber
\begin{aligned}
f_3(N_1)&\triangleq\text{DOF}=-2N_1^3+(2N+\frac{3}{2})N_1^2+(\frac{9}{2}-4N)N_1\\
&\ \ \ \ \ \ \ \ \ \ +4N^2-\frac{3}{2}N-\frac{47}{8}.
\end{aligned}
\end{equation}

For the function $f_3(N_1)$, when the total number of physical sensors $N$ is given, the unknown parameter is only $N_1$,
which is a cubic function $f_3(N_1)$ about $N_1$.
There may be a local maximum for $f_3(N_1)$ within the independent variable interval $[0, N]$.
To discuss the local maximum of $f_3(N_1)$ and obtain the stationary point of $f_3(N_1)$,
the first derivative of $f_3(N_1)$ is derived firstly as follows
\begin{equation}\nonumber
\frac{\partial f_3(N_1)}{\partial N_1}=-6N_1^2+(3+4N)N_1+\frac{9}{2}-4N.
\end{equation}

Therefore, the stationary points of $f_3(N_1)$ can be derived as follows when $\frac{\partial f_3(N_1)}{\partial N_1}=0$
\begin{equation}\nonumber
\begin{aligned}
&N_{11},N_{12}=\frac{3+4N\pm\sqrt{(4N-9)^2+36}}{12}.
\end{aligned}
\end{equation}

Since $N_{11}$ and $N_{12}$ represent the number of sensors, that means they are greater than zero.
According to the properties of quadratic functions,
$\frac{3+4N-\sqrt{(4N-9)^2+36}}{12}<0$,
which is not satisfied the condition,
thus only $N_{11}$ is considered for the local maximum of $f_3(N_1)$.

Furthermore, since the first derivative $\frac{\partial f_3(N_1)}{\partial N_1}$ is greater than zero at $[0, N_1]$,
$f_3(N_1)$ strictly monotonically increases at $[0, N_1]$.
On the contrary, since the first derivative $\frac{\partial f(N_1)}{\partial N_1}$ is less than zero at $[N_1,N]$,
$f_3(N_1)$ strictly monotonically decreases at $[N_1,N]$.
Therefore, $f_3(N_1)$ can obtain the maximum value at the stationary point $N_1^*\triangleq \frac{3+4N+\sqrt{(4N-9)^2+36}}{12}$ for the interval $[0,N]$.
Additionally, $N_1^*=\lceil\frac{3+4N+\sqrt{(4N-9)^2+36}}{12}\rfloor$ due to $N_1^*$ representing the number of sensors.
And the maximum DOF is $-2(N_1^*)^3+(2N+\frac{3}{2})(N_1^*)^2+(\frac{9}{2}-4N)N_1^*+4N^2-\frac{3}{2}N-\frac{47}{8}$.

\section{The Proof of Corollary 1}
In order to further investigate the redundancy of TO-SDA(CNA),
the relaxed problem of the redundancy as follows is studied
\begin{equation}\nonumber
\begin{aligned}
&N_1^*=\frac{4N+\sqrt{(4N+15)^2+672}-21}{12},\\
&U_1\sx =\sx (-(N_1^*)^3\sx +\sx (N-\frac{21}{4})(N_1^*)^2 \sx +\sx (6N+\frac{19}{2})N_1^*\sx -\sx \frac{17}{4}\sx -\sx 5N\sx -\sx 1)/2,\\
&R_T^1(N)=\frac{\frac{2}{3}N^3+\frac{N^2}{2}-\frac{N}{6}}{U_1}.
\end{aligned}
\end{equation}

And the upper and lower bounds of redundancy of TO-SDA(CNA) when the number of physical sensors
varies from 2 to infinity is analyzed in detail as follows.

Firstly, when there are two or three physical sensors for TO-SDA, the redundancies of TO-SDA(CNA) is
\begin{equation}\nonumber
R_T^1(2)=7,\ \ \ \  R_T^1(3)=2.4789.
\end{equation}

When the number of physical sensors $N\rightarrow \infty$, the redundancy of TO-SDA(CNA) is
\begin{equation}\nonumber
\begin{aligned}
\lim_{N\rightarrow \infty} R_T^1(N)&=\lim_{N\rightarrow \infty}\frac{\frac{2}{3}N^3+\frac{N^2}{2}-\frac{N}{6}}{U_1}=9.\\
\end{aligned}
\end{equation}

Further the monotonicity of $R_T^1(N)$ is analyzed, it can be concluded that within $N\in[2,+\infty)$, $R_T^1(N)$ increases monotonically.
Therefore, the upper and lower bounds of the redundancy $R_T^1$ with the number of physical sensors varying from 2 to infinity are
\begin{equation}\nonumber
\begin{aligned}
2.4789\leq R_T^1\leq 9.
\end{aligned}
\end{equation}

Similarly, it can be proven that the upper and lower bounds of redundancies for TO-SDA(SCAN) and TO-SDA(TNA-II) are shown as Corollary 1.

}

\end{document}